\documentclass[sigconf, nonacm]{acmart}

\newcommand\vldbdoi{XX.XX/XXX.XX}
\newcommand\vldbpages{XXX-XXX}
\newcommand\vldbvolume{14}
\newcommand\vldbissue{1}
\newcommand\vldbyear{2020}
\newcommand\vldbauthors{\authors}
\newcommand\vldbtitle{\shorttitle} 
\newcommand\vldbavailabilityurl{URL_TO_YOUR_ARTIFACTS}
\newcommand\vldbpagestyle{plain} 

\usepackage{algorithm}
\usepackage{algorithmicx}
\usepackage{algpseudocode}   
\usepackage{bbm}
\usepackage{subfigure}
\usepackage{url}
\usepackage{paralist}
\usepackage{multirow}
\usepackage{diagbox}

\begin{document}
\title{DUMP: A Dummy-point-based Local Differential Privacy Enhancement Approach under the Shuffle Model}

\author{Xiaochen Li}
\affiliation{%
  \institution{Zhejiang University}
}
\email{xiaochenli@zju.edu.cn}

\author{Weiren Liu}
\affiliation{%
  \institution{Alibaba Group}
}
\email{weiran.lwr@alibaba-inc.com}

\author{Hanwen Feng}
\affiliation{%
\institution{Alibaba Group}
}
\email{fenghanwen.fhw@alibaba-inc.com}

\author{Kunzhe Huang}
\affiliation{%
\institution{Zhejiang University}
}
\email{hkunzhe@zju.edu.cn}

\author{Yuke Hu}
\affiliation{%
\institution{Zhejiang University}
}
\email{yukehu@zju.edu.cn}

\author{Jinfei Liu}
\affiliation{%
  \institution{Zhejiang University}
}
\email{jinfeiliu@zju.edu.cn}

\author{Kui Ren}
\affiliation{%
\institution{Zhejiang University}
}
\email{kuiren@zju.edu.cn}

\author{Zhan Qin}
\affiliation{%
\institution{Zhejiang University}
}
\email{qinzhan@zju.edu.cn}

\begin{abstract}
The shuffle model is recently proposed to address the issue of severe utility loss in Local Differential Privacy (LDP) due to distributed data randomization. 
    In the shuffle model, a shuffler is utilized to break the link between the user identity and the message uploaded to the data analyst. 
    Since less noise needs to be introduced to achieve the same privacy guarantee, following this paradigm, the utility of privacy-preserving data collection is improved.

    In this paper, we propose DUMP (\underline{DUM}my-\underline{P}oint-based), a new framework for privacy-preserving histogram estimation in the shuffle model. 
    In DUMP, dummy messages are introduced on the user side, creating an additional \emph{dummy blanket} to further improve the utility of the shuffle model.
    In addition, this new privacy-preserving mechanism also achieves LDP privacy amplification effect to the user uploaded data against a curious shuffler. 
    We instantiate DUMP by proposing two protocols: pureDUMP and mixDUMP, both of which achieving theoretical improvements on accuracy and communication complexity comparing with existing protocols. 
    A comprehensive experimental evaluation is conducted to compare the proposed two designs under DUMP with the existing protocols.
    Experimental results on both synthetic and real-world datasets confirm our theoretical analysis and show that the proposed protocols achieve better utility than existing protocols under the same privacy guarantee.
\end{abstract}

\maketitle

\pagestyle{\vldbpagestyle}
\begingroup\small\noindent\raggedright\textbf{PVLDB Reference Format:}\\
\vldbauthors. \vldbtitle. PVLDB, \vldbvolume(\vldbissue): \vldbpages, \vldbyear.\\
\href{https://doi.org/\vldbdoi}{doi:\vldbdoi}
\endgroup
\begingroup
\renewcommand\thefootnote{}\footnote{\noindent
This work is licensed under the Creative Commons BY-NC-ND 4.0 International License. Visit \url{https://creativecommons.org/licenses/by-nc-nd/4.0/} to view a copy of this license. For any use beyond those covered by this license, obtain permission by emailing \href{mailto:info@vldb.org}{info@vldb.org}. Copyright is held by the owner/author(s). Publication rights licensed to the VLDB Endowment. \\
\raggedright Proceedings of the VLDB Endowment, Vol. \vldbvolume, No. \vldbissue\ %
ISSN 2150-8097. \\
\href{https://doi.org/\vldbdoi}{doi:\vldbdoi} \\
}\addtocounter{footnote}{-1}\endgroup

\ifdefempty{\vldbavailabilityurl}{}{
\vspace{.3cm}
\begingroup\small\noindent\raggedright\textbf{PVLDB Artifact Availability:}\\
The source code, data, and/or other artifacts have been made available at \url{\vldbavailabilityurl}.
\endgroup
}

\section{Introduction}
Differential privacy is widely accepted, by both academia and industry  \cite{erlingsson2014rappor,team2017learning,ding2017collecting,wang2019answering}, as the standard privacy definition for statistical analysis on databases containing sensitive user information. In the last decades, the community has devoted 
a number of efforts to designing practical mechanisms achieving better utility while providing \textit{reliable} privacy guarantee for users, and most of them can be categorized into the central model \cite{geng2015staircase, ghosh2012universally} or local model \cite{kairouz2014extremal, holohan2017optimal}.  In the central model, the analyst, who collects data from individuals,  is responsible to add noise to the aggregated result such that the noisy result will not leak too much information about each individual.  In contrast, the local model enables each user to perturb his data before submitting it to an analyst and thus provides a more reliable privacy guarantee even when the analyst is \textit{curious}. 
However, mechanisms achieving local differential privacy (LDP) only have limited utility and turn out ill-suited to practice \cite{bittau2017prochlo}.
It is known that the estimation error is at least $\Omega(\sqrt{n})$ for a wide range of natural problems in the local model ($n$ is the number of users), while an error of just $O(1)$ can be achieved in the central model\cite{duchi2013local}.

The above dilemma motivates the study of new approaches that prevent third parties from violating users' privacy while achieving reasonable utility; the shuffle model, initiated by Bittau et al. \cite{bittau2017prochlo}, is emerging as a promising one.  This model introduces an independent third party called \textit{shuffler} who is responsible for shuffling all (permuted) messages received from users before passing them to the analyst\footnote{In practice, various cryptographic tools like nest encryptions \cite{bittau2017prochlo} can be employed to ensure that the shuffler cannot learn the content of these reports.}. Regarding accuracy, for queries that are independent of the sequence, this shuffling will not introduce extra utility loss. Regarding privacy, intuitively,  the analyst cannot link each report in a shuffled database to an individual, and thus the privacy against analysts can be amplified. Or, less ``noise" shall be added by users to messages for the same privacy, which in turn improves the utility. Though similar (or even better) results could be achieved using heavy cryptographic primitives like homomorphic encryptions and multiparty computation protocols \cite{lazar2018karaoke, roy2020crypt}, shuffle-model approaches are more popular in the industry since they are considerably efficient and \textit{implementation-friendly}!  In recent years, several shuffle-model protocols have been presented for various tasks (e.g., frequency estimation \cite{balcer2019separating}, data summation \cite{ghazi2019scalable}, and range query \cite{ghazi2019private}) over either numerical data or categorical data.

We focus on the problem of \textit{histogram estimation}, the fundamental task over categorical data, in the shuffle model. Many recent works \cite{wang2020improving, ghazi2019power, balcer2019separating, balle2019privacy,ghazi2020private2} have aimed to design truly practical mechanisms for this problem, and they indeed have different tradeoffs. 
In terms of efficiency and implementation-friendliness, the generalized random response (GRR) based mechanisms\cite{balle2019privacy}, in which a user just replaces his data $x$ by a uniformly random value $r$ in the data domain $\mathbb{D}$ with certain probability $q$, turns out to be optimal: (1) only a single message shall be sent per user, and (2) it directly uses uniform randomness, which is naturally provided by operating systems via standard means. We note the feature (2) is highly desirable in practice: although random variables with many other distributions can be calculated from uniformly random variables, 
the inexact calculation over finite precision values cannot produce randomness exactly having the specified distribution, and the implemented privacy mechanism may leak significant information of the original data (see \cite{DBLP:conf/ccs/Ilvento20,DBLP:conf/ccs/Mironov12,DBLP:conf/innovations/BalcerV18,DBLP:journals/tcs/GazeauMP16}). More interestingly, Balle \textit{et al.} \cite{balle2019privacy} proposed the \textit{privacy blanket} intuition for GRR-based protocols, which gives a thoroughly theoretical explanation on how (and how much) shuffling enhances the privacy, in a conceptually simple manner.

Nonetheless, the estimation error of GRR-based protocols  
will become unexpectedly large as the domain size $|\mathbb{D}|$ grows; the privacy guarantee they can provide is also subjected to the number of users and the data domain size and might not be enough for applications. Subsequent works \cite{wang2020improving, ghazi2019power, balcer2019separating,ghazi2020private2}  improve the accuracy and achieve errors independent of $|\mathbb{D}|$ (under the same privacy guarantee) by letting users randomize their data in rather complex ways (e.g., encoding data using hash functions \cite{wang2020improving} or Hadamard matrix \cite{ghazi2019power}, adding noises from negative binomial distributions \cite{ghazi2020private2}, etc.); some of them \cite{ghazi2019power,ghazi2020private2} also enables high privacy guarantee. Yet they indeed sacrifice efficiency and/or implementation-friendliness (detailed comparison will be presented later). 

In this work, we present shuffle-model mechanisms for histogram estimation, which are almost as efficient and implementation-friendly as GRR-based protocols \cite{balle2019privacy}, enable high privacy guarantee, and have competitive accuracy with \cite{wang2020improving, ghazi2019power, balcer2019separating,ghazi2020private2}. Moreover, we give an intuition that clearly explains how and how much our mechanisms provide privacy (mimicking the \textit{privacy blanket} intuition); we believe that such intuition can assist in the design and analysis of future mechanisms.

\medskip
\noindent \textbf{Our Contributions.}
We overview our contributions in more detail below.

\noindent\underline{\it Dummy blanket \& DUMP framework. } 
We introduce the concept of \textit{dummy blanket}, which enables a very simple method of enhancing user's privacy (particularly, in the shuffle model) and significantly outperforms the \textit{privacy blanket} \cite{balle2019privacy} in the accuracy.  Recall that privacy blanket provides an insightful view of the GRR mechanism: if each user reports the true value $v$ with probability $p$ and a uniformly sampled $v'$ with probability $q=1-p$, then, equivalently, there is expected to be $nq$ uniformly random values among $n$ reports. After shuffling, the victim data $v^*$ is mixed with the $nq$ points, and these points will serve as a ``blanket" to hide $v^*$. Our \textit{dummy blanket} intuition indeed generalizes the privacy blanket. Specifically, we observe that privacy essentially comes from these uniformly random points and is determined by the number of them, and we do not need to insist that these points are contributed by data randomization. Instead, we allow users to directly add uniformly random points (called \textit{dummy points}) before reporting to the shuffler. Similar to the case of privacy blanket, these dummy points will also serve as a ``blanket" to protect the victim!

Dummy blanket allows us to achieve better accuracy at the cost of slightly increasing communication overhead. To see this, we first show the privacy guarantee is mostly determined by the total number of dummy points. Then, for achieving certain privacy (that needs $s$ dummy points), the privacy blanket intuition requires to replace $s$ true values, while we just add $s$ dummy points. Since original values are preserved, the estimation result of the dummy blanket will be more accurate. Formally, we show the estimation error introduced by dummy blanket is independent of the domain size $|\mathbb{D}|$ and always concretely smaller than that of privacy blanket (see Sect.\ref{subsec:pureDUMP}). Regarding communication overhead, note that only the total number of dummy points matters, and thus the amortized communication overhead will be close to a single message when $n$ is large.  
Furthermore,  to some extent, the dummy blanket intuition could explain why previous works \cite{balcer2019separating,wang2020improving} improved accuracy: their data randomizers implicitly generate dummy points, though the methodology of adding dummy points and its advantage had never been formally studied.

We incorporate the dummy blanket intuition into the existing shuffle-model framework and present a general DUMP (\underline{DUM}my-\underline{P}oint-based) framework for histogram estimation. In DUMP, a user first processes his data using a data randomizer (as he does in the conventional shuffle-model framework), then invokes a dummy-point generator to have a certain amount of dummy points, and finally submits the randomized data together with these dummy points to the shuffler.
From the dummy blanket intuition, adding dummy points could enhance privacy without significantly downgrading the accuracy. Note that most existing locally private protection mechanisms (e.g., Randomized Response \cite{warner1965randomized}, RAPPOR \cite{erlingsson2014rappor}, OLH \cite{wang2017locally}) could be formulated as one of the special cases of the data randomizer, and thus the framework admits diverse instantiations which can give different tradeoffs for different purposes.

\smallskip\noindent\underline{\it Instantiations: pureDUMP \& mixDUMP. } We first instantiate DUMP with an empty data randomizer and obtain the pureDUMP protocol. In pureDUMP, the privacy guarantee is fully provided by dummy points, and thus the analysis of pureDUMP could quantify properties of the dummy blanket. Assume each user adds $s$ dummy points and the data domain $\mathbb{D}$ has size $k$. We prove that  (1) pureDUMP satisfies $(\epsilon,\delta)$-differential privacy for $\delta\leq 0.2907$ and $\epsilon=\sqrt{14k\cdot \ln (2/\delta)/(ns-1)}$; (2) its expected mean squared error (eMSE) of frequency estimation is $s(k-1)/(nk^2)$.  So the eMSE is $O(\log(1/\delta)/(\epsilon^2n^2))$, surely independent of $k$, and only increases the error in the central model by $\log (1/\delta)$ that is logarithmic in $n$; the number of messages per user is $1+O(k\log (1/\delta)/(n\epsilon^2))$ that will be close to $1$ as $n$ grows. Moreover, by increasing $s$, $\epsilon$ can be arbitrarily close to 0 as wishes. Meanwhile, we prove that pureDUMP even satisfies local differential privacy, which is meaningful when the shuffler and the analyst collude.

We then present mixDUMP, which utilizes GRR as the data randomizer in the DUMP framework. mixDUMP leads to a new tradeoff between communication overhead and accuracy. Particularly, since privacy can also be provided by GRR, fewer dummy points are needed. We show that mixDUMP has the smallest communication overhead (compared with all existing multi-messages mechnisms) while its eMSE is still competitive (though slightly larger than that of pureDUMP). 



Finally, notice that privacy is determined by the total number of dummy points; when $n$ is large, it may not be necessary to ask each user to send even a single dummy point. Therefore we investigate the setting that each user has a certain probability to generate a fixed number of dummy points. Meanwhile, we find that this setting is interesting for the case that not all users are willing to utilize dummy-points-based privacy protection mechanism due to the limited network bandwidth or computation capability of their devices. We formally prove that pureDUMP and mixDUMP can still provide enhanced privacy protection of $(\epsilon, \delta)$-differential privacy in this setting. This result allows having smaller communication overhead and better accuracy when $n$ is large.

\smallskip\noindent\underline{\it Implementations \& Comparisons.} We analyze and implement all the existing related protocols \cite{balle2019privacy,wang2020improving, ghazi2019power, balcer2019separating,ghazi2020private2}  (at least to the best of our knowledge) for comprehensive comparisons. 
We conduct experiments on both synthetic and real-world datasets; we are the first (also to the best of knowledge) to report \cite{balcer2019separating, ghazi2019power} on large-size experiments.  
With the same privacy guarantee, the results show that (1) pureDUMP and mixDUMP have the minimal communication overhead compared with all existing multi-message protocols (\cite{ghazi2019power, balcer2019separating,ghazi2020private2}); (2) pureDUMP achieves the best accuracy compared with \cite{balle2019privacy,wang2020improving, ghazi2019power, balcer2019separating}, while the only exception \cite{ghazi2020private2}, achieving near-optimal accuracy, has much higher communication-overhead and has to use randomness from negative-binomial distributions whose implementation is indeed challenging and vulnerable in practice. 
Details follow.
\begin{itemize}[-]
	\item \textit{Efficiency.} The communication overhead of both pureDUMP and mixDUMP decreases with the increase of the numbe of users, while \cite{ghazi2019power,balcer2019separating} do not have this feature; the expected communication overhead of \cite{ghazi2020private2} turns out to be concretely larger.
	In our experiments, for example, over the real-world dataset \textit{Ratings} \cite{ratings-web} with $\epsilon=1$ and $\delta=10^{-6}$, pureDUMP (and mixDUMP) only require around $0.8$ (and $0.5$) expected extra messages per user, while \cite{ghazi2019power,balcer2019separating} and \cite{ghazi2020private2} need around $10^2$, $10^3$ and $10^4$ extra messages respectively.  On the other hand, single-message protocols \cite{balle2019privacy,wang2020improving} are less competitive in accuracy and/or subjected to upper-bounds on privacy guarantee (namely, the lower-bounds on $\epsilon$). Particularly, for Ratings, \cite{balle2019privacy} \textit{cannot }even achieve $(\epsilon,\delta)$-privacy with $\epsilon\leq 1$. 
	\item \textit{Accuracy.} Our experiements show that pureDUMP and mixDUMP enjoy better accuracy than all existing protocols except \cite{ghazi2020private2}. Particularly, in the experiments over Ratings with $\epsilon=1$ and $\delta=10^{-6}$, the eMSE of pureDUMP and that of \cite{ghazi2020private2} are in $[10^{-10},10^{-9} )$, while the eMSEs of others are close to $10^{-8}$.
	\item \textit{Implementation-friendliness.}  pureDUMP and mixDUMP are extremely implementation-friendly, as they just generate a fixed amount of uniformly random values with fixed probability and do not involve any cumputation that shall be approximated over finite-precision devices (e.g., base-$e$ logarithm $\ln(\cdot)$). In contrast,  \cite{ghazi2020private2} uses randomness with the negative binomial distribution $\mathsf{NB}(r,p)$ s.t. $r=1/n$ and $p=e^{-\epsilon}$.  \footnote{For $k\in \mathbb{Z}$, and $X\leftarrow \mathsf{NB}(r,p)$, $\Pr[X=k]=\tbinom{k+r-1}{k}(1-p)^rp^k$.}  Notice that $r=1/n$ is not an integer, so we cannot sample from this distribution by naively running independent events that have succesful probability $p$ untill seeing $r$ failures. We may have to follow the common approach of sampling negative binomial distributions: first sample $\lambda$ from a $\Gamma$ distribution $\Gamma(r, p/(1-p))$ and then sample from a Poisson distribution $\mathsf{Poisson}(\lambda)$. Note that sampling from Poisson distributions involves $\ln(\cdot)$ computation, and Mironov \cite{DBLP:conf/ccs/Mironov12} showed that DP meachnisms implementing $\ln(\cdot)$ over double-precision floating-point numbers cannot provide the claimed privacy protection!  Other sampling methods for  $\mathsf{NB}(r,p)$ might be possible but to be explored.
\end{itemize}

\section{Background} \label{sect:background}
\medskip\noindent{\bf Central Differential Privacy. }
The initial setting of differential privacy is defined in central model \cite{dwork2008differential}, in which a trusted analyst collects sensitive data from users, gathers them into a dataset, and works out the statistical results.
The statistical results are perturbed by the analyst before being released.
Intuitively, differential privacy protects each individual's privacy by requiring any single element in the dataset has limited impact on the output.

\begin{definition}\label{def:differential-privacy} (\emph{Central Differential Privacy, CDP}). 
	An algorithm $\mathcal{M}$ satisfies ($\epsilon,\delta$)-DP,where $\epsilon, \delta \geq 0$, if and only if for any neighboring datasets $\mathcal{D}$ and $\mathcal{D'}$ that differ in one element, 
	and any possible outputs $\mathcal{R} \subseteq Range(\mathcal{M})$, we have  
	\[\Pr \left[ \mathcal{M}(\mathcal{D}) \in \mathcal{R}\right] \leq e^{\epsilon} \Pr \left[ \mathcal{M}(\mathcal{D'}) \in \mathcal{R} \right] + \delta\]
\end{definition}

Here, $\epsilon$ is called the \emph{privacy budget}. 
The smaller $\epsilon$ means that the outputs of $\mathcal{M}$ on two adjacent datasets are more similar, and thus the provided privacy guarantee is stronger. 
$\delta$ is generally interpreted as $\mathcal{M}$ not satisfying $\epsilon$-DP with probability $\delta$. As usual, we consider the neighboring datasets $\mathcal{D}$ and $\mathcal{D'}$ having the same number of elements (aka. bounded differential privacy). For simplicity, we assume w.o.l.g. that $\mathcal{D}$ and $\mathcal{D'}$ differ at the $n$-th element, and we denote them by 
 $\{x_1, x_2, \ldots, x_n\}$ and $\{x_1, x_2, \ldots, x'_n\}$.

\medskip\noindent{\bf Local Differential Privacy.}
The central model assumes the analyst is fully trusted and does not ensure privacy against curious analysts, which is not practical in many scenarios. In constrast, the local model \cite{kasiviswanathan2011can} allows  
each user to locally perturb his data using a randomization algorithm $\mathcal{M}$ before sending it to the untrusted analyst so that the analyst cannot see the raw data. We say $\mathcal{M}$ satisfies local differential privacy (LDP) if for different data $v$ and $v'$, the distribution of $\mathcal{M}(v)$ and that of $\mathcal{M}(v')$ are close. 
\begin{definition}\label{def:local-differential-privacy} (\emph{Local Differential Privacy, LDP}). 
	An algorithm $\mathcal{M}$ satisfies $(\epsilon,\delta)$-LDP, where $\epsilon, \delta \geq 0$, if and only if for any different elements $v, v' \in \mathbb{D}$ , and any possible outputs $\mathcal{R} \subseteq Range(\mathcal{M})$, we have 
	\[\Pr \left[\mathcal{M}(v) \in \mathcal{R}\right] \leq e^{\epsilon} \Pr \left[\mathcal{M}(v') \in \mathcal{R} \right] + \delta\] 
\end{definition}
$\delta$ is typically set as $0$, and $(\epsilon, 0)$-LDP can be denoted by $\epsilon$-LDP.

\medskip\noindent{\bf (Generalized) Randomized Response.}
Randomized Response (RR) \cite{warner1965randomized} is considered as the first LDP mechanism that supports binary response. It allows each user to provide a false answer with certain probability and thus provides plausible deniability to users. RR mechanism can be extended to support multi-valued domain response, known as Generalized Randomized Response (GRR) \cite{wang2017locally}.
Denote $k$ as the size of the domain $\mathbb{D}$. 
Each user with private value $v \in \mathbb{D}$ reports the true value $v' = v$ with probability $p$ and reports a randomly sampled value $v' \in \mathbb{D}$ where $v' \neq v$ with probability $q$.  By standard analysis, it is easy to see that, GRR satisfies $\epsilon$-LDP when
\begin{equation}
	\label{eq:rr}
	\left\{
	\begin{aligned}	
		p & = \frac{e^{\epsilon}}{e^{\epsilon} + k - 1}, \\
		q & = \frac{1}{e^{\epsilon} + k - 1}.
	\end{aligned}
	\right. 
\end{equation}
%

\medskip\noindent{\bf Shuffle Model.}
The shuffle model follows the \textsf{Encode}, \textsf{Shuffle}, \textsf{Analyze} (ESA) architecture proposed by Bittau et al. \cite{bittau2017prochlo}.
A shuffler is inserted between users and the analyst to break the link between the user identity and the message.
Following  \cite{cheu2019distributed}, we describe a shuffle-model protocol by the following three algorithms.
\begin{enumerate}[-]
	\item $R$: $\mathcal{X}\rightarrow\mathcal{Y}^m$. The local randomized encoder $R$ takes a single user's data $x_i$ as input and outputs a set of $m$ messages $y_{i,1},\ldots,y_{i,m}\in\mathcal{Y}$.
	\item$S$: $(\mathcal{Y}^{m})^n\rightarrow\mathcal{Y}^{mn}$. The shuffler $S$ takes $n$ users' output messages of $R$ as input and outputs these messages in a uniformly random order.
	\item $A$: $\mathcal{Y}^{mn}\rightarrow\mathcal{Z}$. The analyst $A$ that takes all the output messages of S as input and runs some analysis functions on these messages.
\end{enumerate}
Such a protocol is called a \emph{single-message} protocol when $m=1$ or a \emph{multi-message} protocol when $m>1$. 


\begin{table}[t]
	\small
	\centering
	\begin{tabular}{|c|c|}
		\hline
		\textbf{Variable}&\textbf{Description}\\
		\hline
		$\mathbb{D}$&Finite and discrete domain \{1, 2 , \ldots, k\}\\
		\hline
		$k$&Size of domain $\mathbb{D}$\\
		\hline
		$\mathcal{D}$&Set of users' input values\\
		\hline
		$n$&Number of users\\
		\hline
		$\mathcal{S}$&Set of dummy points\\
		\hline
		$|\mathcal{S}|$&Total number of dummy points\\
		\hline
		$s$&Number of dummy points sent per user\\
		\hline
		$x_i\in\mathbb{D}$&Input value from the $i^\text{th}$ user\\
		\hline
		$y_{i,j}\in\mathbb{D}$&The $j^\text{th}$ value sent from the $i^\text{th}$ user\\
		\hline
		$\vec z$&The estimated frequency of elements in $\mathbb{D}$\\
		\hline
		$\mathsf{Ber}(\cdot)$&Sample a random binary value under Bernoulli distribution\\
		\hline
		$\mathsf{Unif}([k])$&Uniformly sample a random value from $\{1,2,\ldots ,k\}$\\
		\hline
	\end{tabular}
	\setlength{\abovecaptionskip}{-3ex} 
	\setlength{\belowcaptionskip}{0ex}
	\caption{Important Notations}
	\label{tab:compare_pro}
\end{table}

\section{Histogram Estimation in the shuffle model}
\subsection{Problem Statement}
\noindent\textbf{System Setting.}
The task is performed in the shuffle model. 
The system involves an analyst, a shuffler, and $n$ users.
Specifically, the analyst knows the data domain $\mathbb{D}$.
We assume the size of $\mathbb{D}$ is $k$, i.e., $\mathbb{D}=\{1,2,\ldots,k\}$.
The $i^{\text{th}}$ user has a single data $x_i\in\mathbb{D}$.
All users send perturbed data to the shuffler, and the shuffler scrambles the order before transmitting the data to the analyst.
For the $j^{\text{th}}$ item in $\mathbb{D}$, the analyst estimates frequency $\vec z[j]$ (the proportion of users who has a certain value), i.e., $\vec z[j]=\frac{1}{n}\sum_{i=1}^{n}\mathbbm{1}_{x_i=j}$. Here $\mathbbm{1}_{x_i=j}=1$ if $x_i=j$; $\mathbbm{1}_{x_i=j}=0$ otherwise.

\noindent\textbf{Threat Model.}  Similar to most shuffle-model works, we consider privacy against curious analysts. We concentrate on the standard case, in which the analyst does not collude with shufller. We say a shuffle-model protocol $P=(R,S,A)$ (see Section \ref{sect:background}) satisfies $(\epsilon,\delta)$-DP against curious analysts, if the algorithm $S\circ \cup_i R(x_i)=A(S(R(x_1),\ldots,R(x_n)))$ satisfies $(\epsilon,\delta)$-DP in the central model (see Definition \ref{def:differential-privacy}). We also study the case that the analyst colludes with the shufller such that he has a direct access to users' reports. In this case, the differential privacy of $P$ against the collusion is equivalent to the local differential privacy of $R$ (see Definition \ref{def:local-differential-privacy}).



Following \cite{balle2019privacy}, throughout this paper we consider powerful adversaris who have the same background knowledge, i.e., all other users' values except one victim's; thus our mechanisms are also secure against any weaker adversaries.
All privacy analyses are under the assumption of bounded differential privacy, thus we do not need to consider the case that some users dropout.

\subsection{Existing Related Protocols}
\label{sub:relta histogram_estimation}
There are several existing works on histogram estimation in the shuffle model.
To the best of our knowledge, the first protocol is proposed by Balle et al. \cite{balle2019privacy}, in which each user randomizes sensitive data with GRR mechanism.
Then, other mechanisms are utilized as randomizers to improve the utility of protocols.
Ghazi et al. \cite{ghazi2019power} propose two \emph{multi-message} protocols: private-coin and public-coin.
The private-coin utilizes Hadamard Response as the data randomizer.
The public-coin is based on combining the Count Min data structure \cite{cormode2005improved} with GRR \cite{wang2017locally}.
The user in the truncation-based protocol proposed by Balcer et al. \cite{balcer2019separating} sends sensitive data with at most $k$ (size of data domain) non-repeated data sampled from $\mathbb{D}$ to the shuffler.
However, the truncation-based protocol may suffer from an undesirable communication load, i.e., each user needs to send thousands of messages when there are thousands of candidate values in $\mathbb{D}$.
The SOLH proposed by Wang et al. \cite{wang2020improving} utilizes Optimal Local Hash (OLH) as the data randomizer to improve the utility when the size of $\mathbb{D}$ is large.
The correlated-noise protocol proposed by Ghazi et al. \cite{ghazi2020private2} adds noises from negative binomial distributions, which achieves accuracy that is arbitrarily close to that of CDP.

All the above existing protocols seek for more suitable randomizers to improve the utility of the shuffle model.
Different from existing works, we try to introduce a new mechanism to provide privacy protection.
Our proposed method gets more privacy benefits than existing approaches, which further improves the utility of histogram estimation in the shuffle model.

\section{DUMP Framework}
\subsection{Framework Overview}
There are three parties in DUMP framework: user side, shuffler, and analyst.
Figure \ref{fig:framework_DUMP} shows the interactions among them.

\noindent\textbf{User Side.}
There are two types of protection procedures on the user side.
One is the data randomizer.
Most existing LDP mechanisms could be categorized as one of special cases of the data randomizer such as Randomized Response (RR), Unary Encoding (UE), Optimal Local Hash (OLH), and Hadamard Response (HR).
Hence, the output of the data randomizer $R(x_i)$ could be in different forms such as element, vector, or matrix.
Besides, the data randomizer directly outputs the user's sensitive data when no privacy protection is provided by the data randomizer, i.e., $R(x_i)=x_i$.
Another is the dummy-point generator, which provides privacy protection by generating dummy points.
$s$ dummy points are uniformly random sampled from the output space of the data randomizer, where $s$ is calculated from the number of users $n$, the domain size $k$, and privacy parameters.
The output of $R(x_i)$ mixed with $s$ dummy points is then sent to the shuffler.

\noindent\textbf{Shuffler.}
The operation of shuffler in DUMP is same as the existing shuffle model\cite{bittau2017prochlo}.
The shuffler collects all messages from users and removes all implicit metadata (such as IP address, timestamps, and routing paths).
Then all messages are randomly shuffled before sending to the analyst.

\begin{figure}
	\centering
	\includegraphics[scale=0.45]{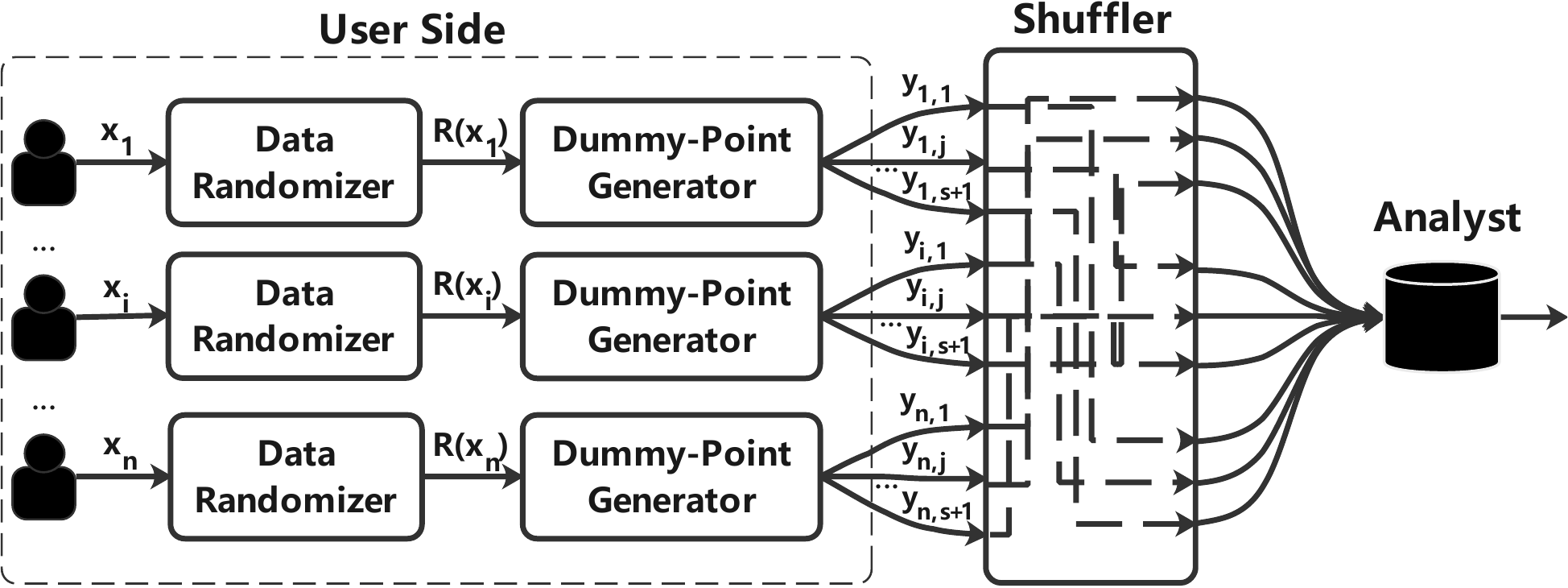}
	\vspace{-3ex}
	\setlength{\belowcaptionskip}{-3ex}
	\caption{The overview of DUMP framework}
	\label{fig:framework_DUMP}
\end{figure}

\noindent\textbf{Analyst.}
The analyst performs histogram estimation on the set of messages received from the shuffler.
It removes the bias caused by dummy points and data randomization to obtain unbiased estimation results.

\medskip
Under DUMP framework, we design two protocols, pureDUMP and mixDUMP, in Section \ref{sec:pro} to show the advantage of dummy points in improving utility.
The dummy points provides privacy protection in both two protocols.
In addition, the mixDUMP uses GRR as the data randomizer, while no protection is provided by the data randomizer in pureDUMP.
We show that the dummy points can improve the accuracy without increasing too much communication overhead under a certain privacy guarantee.
Table \ref{tab:compare_pro} presents our results alongside existing results, in which all protocols are used for $k$-bins histogram estimation in the shuffle model.

\begin{table*}[t]
	\linespread{1.2}
	\small
	\centering
	\resizebox{\textwidth}{!}
	{\begin{tabular}{|c|c|c|c|c|}
			\hline
			\textbf{Protocol}&\textbf{No. Messages per User}&\textbf{No. Bits per Message}&\textbf{Expected Mean Square Error}&\textbf{Condition}\\
			\hline
			privacy-amplification\cite{balle2019privacy}&$1$&$O(\log k)$&$O(\frac{\epsilon^2}{(\epsilon^2(n-1)-k\log (1/\delta))^2}\cdot\log (1/\delta))$&$\sqrt{\frac{1}{n-1}\cdot 14k\log (2/\delta)}<\epsilon\le 1$\\
			\hline
			SOLH\cite{wang2020improving}&$1$&$O(\log k)$&$O(\frac{1}{\epsilon^2n^2-\log (1/\delta)}\cdot\log (1/\delta))$&$\sqrt{\frac{1}{n-1}\cdot 14\log (2/\delta)}<\epsilon\le 1$\\
			\hline
			truncation-based\cite{balcer2019separating}&1+$O(k)$&$O(\log k)$&$O(\frac{1}{\epsilon^2 n^2}\cdot\log (1/\delta))$&$e^{-O(n\epsilon^2)}\le\delta<1/n$\\
			\hline
			private-coin\cite{ghazi2019power}&1+$O(\frac{1}{\epsilon^2}\log\frac{1}{\epsilon\delta})$&$O(\log n\log k)$&$O(\frac{1}{\epsilon^2 n^2}\cdot\log (1/\delta))$&$\epsilon<1$, $\delta<1/\log(k)$\\
			\hline
			public-coin\cite{ghazi2019power}&1+$O(\frac{1}{\epsilon^2}\log^3k\cdot\log(\frac{1}{\delta}\log k))$&$O(\log n+\log\log k)$&$O(\frac{1}{n}\cdot \log (\log(k) / \delta)\cdot \log^2 k)$&$\epsilon<1$, $\delta<1/\log(k)$\\
			\hline
			correlated-noise\cite{ghazi2020private2}&1+$O(\frac{1}{n\epsilon^2}k\cdot\log^2\frac{1}{\delta})$&$O(\log k)$&$O(\frac{1}{\epsilon^2 n^2})$&$\epsilon\le O(1)$\\
			\hline
			mixDUMP&1+$O(\frac{1}{n\epsilon^2}k\cdot\log\frac{1}{\delta})$&$O(\log k)$&$O(\frac{1}{\epsilon^2 n^2}\cdot(1+\frac{k}{e^{\epsilon}-1})^2\cdot\log (1/\delta))$&$\epsilon\le 1$\\
			\hline
			pureDUMP&1+$O(\frac{1}{n\epsilon^2}k\cdot\log\frac{1}{\delta})$&$O(\log k)$&$O(\frac{1}{\epsilon^2 n^2}\cdot\log (1/\delta))$&$\epsilon\le 1$\\
			\hline
	\end{tabular}}
	\setlength{\abovecaptionskip}{-3ex} 
	\setlength{\belowcaptionskip}{0ex}
	\caption{Comparison of results for the histogram estimation. Each user is assumed to hold number of $1$ value from $\mathbb{D}$.}
	\label{tab:compare_pro}
\end{table*}

\subsection{Dummy Blanket}
The \emph{blanket} intuition is first proposed by Balle et al. \cite{balle2019privacy}.
They decompose the probability distribution of the data randomizer's output into two distributions to show the privacy enhancement of the shuffle model.
One distribution is dependent on the input value and the other is independently random.
Consider inputting a dataset $\{x_1,x_2,\ldots,x_n\}$ into the data randomizer.
The output histogram $Y$ can be decomposed as $Y=Y_1\cup Y_2$, 
where $Y_1$ is the histogram formed by independently random data and $Y_2$ is the histogram formed by unchanged sensitive data.
Suppose that the analyst knows $\{x_1,x_2,\ldots,x_{n-1}\}$, and $x_n$ remains unchanged.
The analyst needs to observe $Y_1\cup\{x_n\}$ if it tries to break the privacy of $x_n$ from the dataset sent by the shuffler.
Although $x_n$ is not randomized, $Y_1$ hides $x_n$ like a random blanket.

In this paper, we propose the concept of \emph{dummy blanket} to take further advantage of the shuffle model.
We only consider the uniformly random distributed dummy blanket in this paper.
Whether participating in the randomization or not, each user sends $s$ uniformly random generated dummy points to the shuffler.
The total number of $ns$ dummy points form a histogram $Y_3$ that hides each user's data like a dummy blanket. 
In the aforementioned assumptions, the analyst needs to distinguish $x_n$ from $Y_1\cup Y_3\cup \{x_n\}$.
The dummy blanket provides the same privacy guarantee for each user, while the communication overhead is shared by all users.
It can bring significant enhancement of utility when there are a large number of users.
Moreover, the dummy blanket does not conflict with any randomization mechanism, it can be wrapped outside of the blanket formed by any data randomizer.

\section{Proposed Protocols}
\label{sec:pro}
In this section, we present pureDUMP protocol (Section \ref{subsec:pureDUMP}) and mixDUMP protocol (Section \ref{subsec:mixDUMP}).
For each proposed protocol, we analyze the privacy, accuracy, and efficiency properties.

\subsection{{pureDUMP}: A Pure DUMP protocol}
\label{subsec:pureDUMP}

\noindent\textbf{Design of Algorithms.}
All the privacy guarantee is provided by \emph{dummy blanket} in pureDUMP.
On the user side, instead of randomizing the input data $x_i\in\mathbb{D}$, $s$ dummy points are randomly sampled from the domain $\mathbb{D}$ to hide $x_i$.
Here $s$ is determined by privacy parameters when the number of users and the size of data domain are fixed, and the effect of $s$ on privacy budget $\epsilon$ will be discussed later in privacy analysis.
Then the shuffler collects all $n(s+1)$ messages from users and randomly shuffles them before sending them to the analyst.
From the receiving messages, the analyst can approximate the frequency of each $i\in\mathbb{D}$ with a debiasing standard step.
We omit the description of the shuffling operation as it's same as in the standard shuffle-based mechanisms.
The details of the user-side algorithm and the analyst-side algorithm are given in Algorithm \ref{alg:pure_c} and Algorithm \ref{alg:pure_a}, respectively.

\setlength{\textfloatsep}{1ex}
\begin{algorithm}[t]
	\caption{User-side Algorithm of PureDUMP}
	\label{alg:pure_c}
	{\small{
			\begin{algorithmic}[1]
				\Procedure{P}{$x_i$}:
				\Require {$x_i \in \mathbb{D}$, parameters $n, k, s\in \mathbb{N}^+$}
				\Ensure {Multiset $Y_i\subseteq\mathbb{D}$}
				\For {$j \leftarrow 1 \text{ to }s$} \Comment {Generate dummy points}
				\State {$y_{i,j}\leftarrow \mathsf{Unif}([k])$}
				\EndFor
				\State {Mixes $x_i$ with $\{y_{i,1},\ldots,y_{i,s}\}$ to get $Y_i:=\{y_{i,1},\ldots,y_{i,s+1}\}$}\\
				\Return {$Y_i$}
				\EndProcedure
	\end{algorithmic}}}
\end{algorithm}
\setlength{\floatsep}{1ex}
\setlength{\textfloatsep}{1ex}
\begin{algorithm}[t]
	\caption{Analyst-side Algorithm of PureDUMP}
	\label{alg:pure_a}
	{\small{
			\begin{algorithmic}[1]
				\Procedure{A}{$Y_1,\ldots,Y_n$}:
				\Require {Multiset $\{\{y_{1,1},\ldots,y_{1,s+1}\},\ldots,\{y_{n,1},\ldots,y_{n,s+1}\}\}\subseteq\mathbb{D}$, parameters $n, k, s\in \mathbb{N}^+$}
				\Ensure {Vector $\vec z\in[0,1]^k$} \Comment{Frequency of k elements}
				\For {$m\leftarrow 1 \text{ to }k$}
				\State {$\vec z[m]=\frac{1}{n}\cdot(\sum_{i \in [n] , j \in [s+1]} \mathbbm{1}_{m = y_{i,j}} - \frac{ns}{k})$}
				\EndFor\\
				\Return {$\vec z$}
				\EndProcedure
	\end{algorithmic}}}
\end{algorithm}
\setlength{\floatsep}{1ex}

\noindent\textbf{Privacy Analysis.}
Now we analyze the privacy guarantee of pureDUMP.
First, we show how much the privacy guarantee on the analyst side is provided by \emph{dummy blanket}.
We focus on the regime where the count of each element in the dummy dataset is much smaller than the total number of dummy points, which captures numerous practical scenarios since the size of data domain $\mathbb{D}$ is usually large.
Therefore, our privacy analysis is based on Lemma \ref{lem:assumption}.
In fact, the privacy analyses in \cite{balle2019privacy, wang2020improving} are also need to be based on Lemma \ref{lem:assumption}, but it is ignored.

\begin{lemma}
    \label{lem:assumption}
    Denote $\{a_1,a_2, ..., a_m\}$ are random variables that follow the binomial distribution, and $a_1+a_2+...+a_m=N$.
    For any $i,j\in[1,m], i\neq j$, $a_i$ and $a_j$ are approximately considered as mutually independent when $a_1, a_2,...,a_m<<N$.
\end{lemma}
The proof of Lemma \ref{lem:assumption} is in Appendix \ref{app:assumption}. Then we show the analysis result of pureDUMP in Theorem \ref{the:pureprivacy}.
\begin{theorem}
	\label{the:pureprivacy}
	PureDUMP satisfies ($\epsilon_d, \delta_d$)-DP against the analyst for any $k, s, n \in \mathbb{N}^+$, $\epsilon_d \in (0,1]$, and $\delta_d \in (0,0.2907]$, where      
	\[\epsilon_d= \sqrt{\frac{14 k \cdot \ln{(2/\delta_d)}}{|\mathcal{S}|-1}},\ |\mathcal{S}|=ns\]
\end{theorem}

\begin{proof}
	(Sketch)
	For any two neighboring datasets $\mathcal{D}$ and $\mathcal{D'}$, w.l.o.g., we assume they differ in the $n^\text{th}$ value, and $x_n=1$ in $\mathcal{D}$, $x'_n=2$ in $\mathcal{D'}$.
	Let $\mathcal{S}$ and $\mathcal{S'}$ be two dummy datasets with $|\mathcal{S}|=|\mathcal{S}'|=ns$ dummy points sent by users.
	What the analyst receives is a mixed dataset $\mathcal{V}$ that is a mixture of user dataset $\mathcal{D}$(or $\mathcal{D'})$ and the dummy dataset $\mathcal{S}$(or $\mathcal{S'})$. 
	We denote $\mathcal{V}=\mathcal{D}\cup\mathcal{S}$(or $\mathcal{V}=\mathcal{D'}\cup\mathcal{S'}$).
	Denote $s_j$ as the number of value $j$ in the dummy dataset.
	To get the same $\mathcal{V}$, there should be a dummy point equals to 2 in $\mathcal{S}$ but equals to 1 in $\mathcal{S'}$.
	Thus the numbers of $k$ values in other $|\mathcal{S}|-1$ dummy points are $[s_1-1,s_2,\ldots,s_k]$ in $\mathcal{S}$ and $[s_1,s_2-1,\ldots,s_k]$ in $\mathcal{S'}$.
	Then we have
	\[\setlength\abovedisplayskip{0.5ex}
	\setlength\belowdisplayskip{0.5ex}
	\begin{aligned}
		\frac{\Pr [\mathcal{D}\cup\mathcal{S} = \mathcal{V}]}{\Pr [\mathcal{D'}\cup\mathcal{S'} = \mathcal{V}]} = \frac{(_{s_1 - 1,s_2, \ldots, s_k}^{\ \ \ \ |\mathcal{S}|-1})}{(_{s_1, s_2 - 1, \ldots,s_k}^{\ \ \ \ |\mathcal{S}|-1})} = \frac{s_1}{s_2}
	\end{aligned}\]
	
	As dummy points are all uniformly random sampled from $\mathbb{D}$, both $s_1$ and $s_2$ follow binomial distributions, where  $s_1$ follows $Bin(|\mathcal{S}|-1,\frac{1}{k})+1$ and $s_2$ follows $Bin(|\mathcal{S}|-1,\frac{1}{k})$.
	Denote $S_1$ as $Bin(|\mathcal{S}|-1,\frac{1}{k})+1$ and $S_2$ as $Bin(|\mathcal{S}|-1,\frac{1}{k})$.
	Here, $s_1$ and $s_2$ can be considered as mutually independent variables according to Lemma \ref{lem:assumption}.
	Then the probability that pureDUMP does not satisfy $\epsilon_d$-DP is
	\[\setlength\abovedisplayskip{0.5ex}
	\setlength\belowdisplayskip{0.5ex}
	\begin{aligned}
		\Pr [\frac{\Pr [\mathcal{D}\cup\mathcal{S} = \mathcal{V}]}{\Pr [\mathcal{D'}\cup\mathcal{S'} = \mathcal{V}]}\ge e^{\epsilon_d}]=\Pr [\frac{S_1}{S_2} \geq e^{\epsilon_d}] \leq \delta_d
	\end{aligned}\]
	
	Let $c=E[S_2]=\frac{|\mathcal{S}|-1}{k}$, the $\Pr [\frac{S_1}{S_2} \geq e^{\epsilon_d}]$ can be divided into two cases with overlap that $S_1\ge ce^{\epsilon_d/2}$ and $S_2\le ce^{-\epsilon_d/2}$.
	According to multiplicative Chernoff bound, we have
	\[\setlength\abovedisplayskip{0.5ex}
	\setlength\belowdisplayskip{0.5ex}
	\begin{aligned}
		\Pr [\frac{S_1}{S_2} \geq e^{\epsilon_d}]&\le \Pr [S_1\ge ce^{\epsilon_d/2}]+\Pr [S_2\le ce^{-\epsilon_d/2}]\\
		&\le e^{-\frac{c}{3}(e^{\epsilon_d}-1-1/c)^2}+e^{-\frac{c}{2}(1-e^{-\epsilon_d/2})^2}
		\le \delta_d\\
	\end{aligned}\]
	
	\noindent Both $e^{-\frac{c}{3}(e^{\epsilon_d}-1-1/c)^2}$ and $e^{-\frac{c}{2}(1-e^{-\epsilon_d/2})^2}$ are no greater than $\frac{\delta_d}{2}$. 
	We can deduce that $c$ always satisfies $c\ge \frac{14\ln{(2/\delta_d)}}{\epsilon_d^2}$ when $\epsilon_d\le 1$ and $\delta_{d}\le 0.2907$.
	Substitute $c=\frac{|\mathcal{S}|-1}{k}$ into $c\ge \frac{14\ln{(2/\delta_d)}}{\epsilon_d^2}$, we can get the upper bound of $\epsilon_d$.
	
	The detailed proof is deferred to Appendix \ref{app:pureprivacy}.
\end{proof}

Then, we show that when the shuffler colludes with the analyst, each individual user's data can also be protected in pureDUMP.

\begin{corollary}
	PureDUMP satisfies $(\epsilon_l, \delta_l)$-LDP against the shuffler for any $k,s\in \mathbb{N}^+$, $\epsilon_l \in (0,1]$, and $\delta_l \in (0,0.2907]$, where      
	\[\epsilon_l= \sqrt{\frac{14 k \cdot \ln{(2/\delta_l)}}{s-1}}\]
\end{corollary}

\begin{proof}
	The proof is similar to Theorem \ref{the:pureprivacy}, except that each sensitive data only protected by $s$ dummy points generated by its owner.
	More specifically, LDP protection is provided by \emph{dummy blanket} composed of $s$ dummy points before the shuffling operation.
\end{proof}

\noindent\textbf{Accuracy Analysis.}
We analyze the accuracy of pureDUMP by measuring the Mean Squared Error (MSE) of the estimation results.
\begin{equation}
	\label{eq:mse_the}
	MSE = \frac{1}{k} \sum_{v \in \mathbb{D}} E[(\tilde f_v - f_v)^2]
\end{equation}
\noindent where $f_v$ is the true frequency of element $v \in \mathbb{D}$ in users’ dataset, and $\tilde f_v$ is the estimation result of $f_v$.
We first show $\tilde f_v$ of pureDUMP is unbiased in Lemma \ref{lem:pureuti}.

\begin{lemma}
	\label{lem:pureuti}
	The estimation result $\tilde f_v$ for any $v\in\mathbb {D}$ in pureDUMP is an unbiased estimation of $f_v$.
\end{lemma}
\begin{proof}
	\[\setlength\abovedisplayskip{0.5ex}
	\setlength\belowdisplayskip{0.5ex}
	\begin{aligned}
		E[\tilde f_v] & = E \left[\frac{1}{n}\sum_{i \in [n], j \in [s+1]} \left(\mathbbm{1}_{v = y_{i,j}} - \frac{ns}{k}\right) \right] \\
		& =\frac{1}{n} E\left[\sum_{i \in [n], j \in [s+1]} \mathbbm{1}_{v=y_{i,j}}\right]-\frac{s}{k}\\
		& =\frac{1}{n}(n f_v + \frac{ns}{k})-\frac{s}{k} \\
		& =f_v \\
	\end{aligned}\]
\end{proof}

Since $\tilde f_v$ is an unbiased estimation of $f_v$, the MSE of $\tilde f_v$ in pureDUMP is equal to the variance of $\tilde f_v$ according to Equation (\ref{eq:mse_the}).

\begin{equation}
	\label{eq:mse}
	\setlength\abovedisplayskip{0ex}
	\setlength\belowdisplayskip{0ex}
	\begin{aligned}
		MSE & = \frac{1}{k} \sum_{v \in \mathbb{D}} E[(\tilde f_v-f_v)^2] = \frac{1}{k} \sum_{v \in \mathbb{D}} (Var[\tilde f_v]+[E[\tilde f_v]-f_v]^2) \\
		& =\frac{1}{k}\sum_{v\in \mathbb{D}} Var(\tilde f_v)\\
	\end{aligned}
\end{equation}

Now we calculate the MSE of pureDUMP in Theorem \ref{the:pureuti}.

\begin{theorem} \label{the:pureuti}
	The MSE of frequency estimation in pureDUMP is 
	\[\setlength\abovedisplayskip{0.6ex}
	\setlength\belowdisplayskip{0.5ex}
	\frac{s(k-1)}{n k^2}\].
\end{theorem}

\begin{proof}
	\[\setlength\abovedisplayskip{0.5ex}
	\setlength\belowdisplayskip{0.5ex}
	\begin{aligned}
		Var[\tilde f_v] & =Var\left[\frac{1}{n} \left(\sum_{i \in [n], j \in [s+1]} \mathbbm{1}_{v=y_{i,j}}-\frac{ns}{k}\right)\right]\\
		& =\frac{1}{n^2} Var\left[\sum_{i \in [n], j\in [s+1]} \mathbbm{1}_{v=y_{i,j}}\right]\\
		& =\frac{1}{n^2}(\frac{1}{k} \cdot (1-\frac{1}{k}) \cdot ns)\\
		& =\frac{s(k-1)}{n k^2}\\
	\end{aligned}\]
	
	According to Equation (\ref{eq:mse}), we have
	\[\setlength\abovedisplayskip{0.5ex}
	\setlength\belowdisplayskip{0.5ex}
	\begin{aligned}
		MSE = \frac{1}{k} \sum_{v \in \mathbb{D}} E[(\tilde f_v - f_v)^2] = \frac{1}{k} \sum_{v \in \mathbb{D}} Var(\tilde f_v) = \frac{s(k-1)}{n k^2}
	\end{aligned}\]			
\end{proof}

\noindent\textbf{Efficiency Analysis.}
Now we analyze the efficiency of pureDUMP.
For an input $x_i\in\mathbb{D}$, the output of each user side $P(x_i)$ consists of $s+1$ messages of length $O(\log k)$ bits.
The runtime of the user side $P(x_i)$ is at most $O(s)$.
Moreover, $s$ is at most $O(k\ln(1/\delta)/(n\epsilon^2))$ when pureDUMP satisfies $(\epsilon,\delta)$-DP on the analyst side.
The shuffler transmits $n(s+1)$ messages to the analyst with the space of $O(n s\log k)$.
The runtime of the analyst on input $n(s+1)$ messages is at most $O(kns)$ and its output has space $O(k\log(n(s+1)))$ bits.

According to Theorem \ref{the:pureuti}, the expected MSE of estimation in pureDUMP is $O(\log(1/\delta)/(\epsilon^2n^2))$.
Note that the domain size has almost no effect on the estimation error because more dummy points are generated to offset the error caused by the larger domain. 
When the user side only uses a data randomizer that satisfies $\epsilon_l$-LDP to provide protection, we can deduce from results in \cite{balle2019privacy} that the expected MSE is $O(\epsilon^2\log (1/\delta)/(\epsilon^2(n-1)-k\log (1/\delta))^2)$.
By comparison, the estimation error of pureDUMP is always smaller than the method that only using data randomization on the user side, which proves the \emph{dummy blanket} dominates the randomization in the shuffle model.

\subsection{{mixDUMP}: A DUMP Protocol with GRR}
\label{subsec:mixDUMP}

\noindent\textbf{Design of Algorithms.}
mixDUMP utilizes GRR as the data randomizer to provide privacy guarantee along with dummy points.
The details of the algorithm are given in Algorithm \ref{alg:mix_r}.
Note that GRR mechanism utilized in mixDUMP is a little different from the form in \cite{wang2017locally}.
To analyze the privacy benefit with the idea of \emph{blanket}, we use the technique proposed in \cite{balle2019privacy} to decompose GRR into
\[\setlength\abovedisplayskip{0.5ex}
\setlength\belowdisplayskip{0.5ex}
\begin{aligned}
	\forall_{x_r\in\mathbb{D}} \Pr [GRR(x) = x_r] = (1 - \lambda)\Pr [x=x_r] + \lambda \Pr [\mathsf{Unif}([k])=x_r]
\end{aligned}\]

\noindent where $\Pr [x=x_r]$ is the real data distribution that depends on $x$, and $\Pr [\mathsf{Unif}([k])=x_r]$ is the uniformly random distribution.
With probability $1-\lambda$, the output is true input, and with probability $\lambda$, the output is random sampled from the data domain.
Given the probability of $q=1/(e^{\epsilon_l}+k-1)$ in GRR, it has 
\[\setlength\abovedisplayskip{0.5ex}
\setlength\belowdisplayskip{0.5ex}
\begin{aligned}
	\lambda=kq=\frac{k}{e^{\epsilon_l} +k-1}
\end{aligned}\]

\noindent where $\epsilon_l$ is the privacy budget of LDP.

The procedure of generating dummy points is same as Algorithm \ref{alg:pure_c}.
Same as before, we omit the description of shuffling operation.
The details of the user-side algorithm and the analyst-side algorithm are shown in Algorithm \ref{alg:mix_c} and Algorithm \ref{alg:mix_a}, respectively.

\setlength{\textfloatsep}{1ex}
\begin{algorithm}[t]
	\caption{Data Randomizer of MixDUMP}
	\label{alg:mix_r}
	{\small{
			\begin{algorithmic}[1]
				\Procedure{R}{$x_i$}:
				\Require {$x_i \in \mathbb{D}$, parameters $k\in \mathbb{N}^+$ and $\lambda\in(0,\ 1]$}
				\Ensure {$x_r\in\mathbb{D}$}
				\State {$b\leftarrow \mathsf{Ber}(\lambda)$} \Comment {Randomize data with probability $\lambda$}
				\If {$b=0$}
				\State{$x_r\leftarrow x$}
				\Else 
				\State {$x_r\leftarrow \mathsf{Unif}([k])$}
				\EndIf
				\EndProcedure
	\end{algorithmic}}}
\end{algorithm}
\setlength{\floatsep}{1ex}

\noindent\textbf{Privacy Analysis.}
We show how much the privacy guarantee is provided by both randomization and \emph{dummy blanket} to compare with pureDUMP.
In particular, we analyze the privacy amplification provided by \emph{dummy blanket} locally to distinguish the work of adding dummy points on the shuffler side \cite{wang2020improving}. 
We first show the privacy guarantee on the analyst side in Theorem \ref{the:mixprivacy}.
\begin{theorem}
	\label{the:mixprivacy}
	MixDUMP satisfies ($\epsilon_s, \delta_s$)-DP for any $k, s, n \in \mathbb{N}^+$, $\epsilon_s \in (0, 1]$, $\lambda \in (0, 1]$, and $\delta_s \in (0,0.5814]$, where   
	\[\epsilon_s=\sqrt{\frac{14k \ln{(4/\delta_s)}}{|\mathcal{S}|+(n-1)\lambda-\sqrt{2(n-1)\lambda \ln{(2/\delta_s)}}-1}},\ |\mathcal{S}|=ns\]
	
\end{theorem}

\begin{proof}
	(Sketch)
	The assumption is same as in Theorem \ref{the:pureprivacy}.
	We assume the $n^\text{th}$ value $x_n=1$ in $\mathcal{D}$ and $x'_n=2$ in its neighboring dataset $\mathcal{D'}$.
	Denote $R(\mathcal{D})$ (or $R(\mathcal{D}')$) as randomized users' dataset $\{R(x_1),...,R(x_n)\}$ (or $\{R(x_1),...,R(x_n')\}$).
	If the $n^{\text{th}}$ user randomizes its data, then we can easily get
	\[\setlength\abovedisplayskip{0.5ex}
	\setlength\belowdisplayskip{0.5ex}
	\begin{aligned}
		\frac{\Pr [R(\mathcal{D})\cup \mathcal{S} = \mathcal{V}]}{\Pr [R(\mathcal{D'})\cup\mathcal{S'} = \mathcal{V}]}=1
	\end{aligned}\]
	
	So, we focus on the case that the $n^{\text{th}}$ user does not randomize its data.
	Denote $\mathcal{H}$ as the set of users participating in randomization in Algorithm \ref{alg:mix_r}, and the size of $\mathcal{H}$ is $|\mathcal{H}|$.
	Since the randomized users' data and dummy points are all follows the uniform distribution, mixDUMP can be regarded as pureDUMP with dummy dataset $\mathcal{H}\cup\mathcal{S}$ is introduced.
	From Theorem \ref{the:pureprivacy}, mixDUMP satisfies $(\epsilon_H,\frac{\delta_s}{2})$-DP, we have
	\[\setlength\abovedisplayskip{0.5ex}
	\setlength\belowdisplayskip{0.5ex}
	\begin{aligned}
		\Pr [R(\mathcal{D})\cup\mathcal{S}=\mathcal{V}]\le e^{\epsilon_H}\Pr [R(\mathcal{D'})\cup\mathcal{S'}=\mathcal{V}]+\frac{\delta_s}{2}
	\end{aligned}\]
	where $\epsilon_H=\sqrt{\frac{14k\ln{(4/\delta_s)}}{|\mathcal{H}|+|\mathcal{S}|-1}}$, $\delta_s\le 0.5814$.
	Since $|\mathcal{H}|$ follows $Bin(n-1,\lambda)$, according to Chernoff bound, we can get 
	\[\setlength\abovedisplayskip{0.5ex}
	\setlength\belowdisplayskip{0.5ex}
	\begin{aligned}
		\Pr [|\mathcal{H}|<(1-\gamma)\mu]<\frac{\delta_s}{2}
	\end{aligned}\]
	where $\gamma=\sqrt{\frac{2\ln(2/\delta_s)}{(n-1)\lambda}}$.
	Denote $t=(n-1)\lambda-\sqrt{2(n-1)\lambda ln(2/\delta_s)}$, then we have
	\[\setlength\abovedisplayskip{0.5ex}
	\setlength\belowdisplayskip{0.5ex}
	\begin{aligned}
		&\Pr [R(\mathcal{D})\cup \mathcal{S}=\mathcal{V}]\\
		&\le \Pr [R(\mathcal{D})\cup \mathcal{S}=\mathcal{V} \wedge |\mathcal{H}|\ge t]+\frac{\delta_s}{2}\\
		& \le(\sum_{h\ge t} \Pr [R(\mathcal{D})\cup \mathcal{S}=\mathcal{V}]\Pr [|\mathcal{H}|=h])+\frac{\delta_s}{2}\\
		& \le e^{\sqrt{\frac{14k\ln{(4/\delta_s)}}{t+|\mathcal{S}|-1}}}\Pr [R(\mathcal{D'})\cup \mathcal{S'}=\mathcal{V}]+\delta_s\\
	\end{aligned}\]
	
	The details of this proof is in Appendix \ref{app:mixprivacy}. 
\end{proof}

\setlength{\textfloatsep}{1ex}
\begin{algorithm}[t]
	\caption{User-side Algorithm of MixDUMP}
	\label{alg:mix_c}
	{\small{
			\begin{algorithmic}[1]
				\Procedure{M}{$x_i$}:
				\Require {$x_i \in \mathbb{D}$, parameters $n, k, s\in \mathbb{N}^+$ and $\lambda\in(0,\ 1]$}
				\Ensure {Multiset $Y_i\subseteq\mathbb{D}$}
				\State {$x_r\leftarrow R(x_i)$} \Comment {Randomize real data}
				\State {$Y_i\leftarrow P(x_r)$} \Comment {Mix $x_r$ with $s$ dummy points}\\
				\Return {$Y_i$}
				\EndProcedure
	\end{algorithmic}}}
\end{algorithm}
\setlength{\floatsep}{1ex}

\setlength{\textfloatsep}{1ex}
\begin{algorithm}[t]
	\caption{Analyst-side Algorithm of MixDUMP}
	\label{alg:mix_a}
	{\small{
			\begin{algorithmic}[1]
				\Procedure{A}{$Y_1, \ldots ,Y_n$}:
				\Require {Multiset $\{\{y_{1,1}, \ldots ,y_{1,s+1}\},\ldots,\{y_{n,1},\ldots,y_{n,s+1}\}\}\subseteq\mathbb{D}$, parameters $n, k, s\in \mathbb{N}$, $\lambda\in(0,\ 1]$}
				\Ensure {Vector $\vec z\in[0,1]^k$} \Comment{Frequency of k elements}
				\For {$m\leftarrow 1 \text{ to }k$}
				\State {$\vec z[m]=\frac{1}{n}\cdot(\sum_{i \in [n] , j \in [s+1]} \mathbbm{1}_{m = y_{i,j}} - \frac{n(\lambda+s)}{k})/(1-\lambda)$}
				\EndFor\\
				\Return {$\vec z$}
				\EndProcedure
	\end{algorithmic}}}
\end{algorithm}
\setlength{\floatsep}{1ex}

Then, we analyze the LDP privacy guarantee provided by mixDUMP when the shuffler colludes with the analyst.
We show that our design different from \cite{wang2020improving} can obtain privacy benefits from dummy points before shuffling operation, which enables users to be less dependent on trust in the shuffler.
When the data randomizer satisfies $\epsilon_l$-LDP, the introduced dummy points can achieve an enhanced LDP with the privacy budget smaller than $\epsilon_l$.
The theoretical analysis is shown in Theorem \ref{the:shuffle_privacy}.

\begin{theorem}
	\label{the:shuffle_privacy}
	If the data randomizer satisfies $\epsilon_l$-LDP, after introducing $s$ dummy points, $(\epsilon_r, \delta_r)$-LDP is achieved for any $\frac{k}{s}\le min(\frac{1}{2\ln(1/\delta_r)},(\sqrt{c^2+\frac{4c}{\sqrt{2\ln(1/\delta_r)}}}-c)^2)$,
	$c=\sqrt{2\ln(1/\delta_r)}\cdot\frac{e^\epsilon_l(e^{\epsilon_l}+1)}{s(e^{\epsilon_l}+1)+e^{\epsilon_l}}$,
	where 
	\[\setlength\abovedisplayskip{0.5ex}
	\setlength\belowdisplayskip{0.5ex}
	\epsilon_r=\ln{(\frac{k}{1-\sqrt{2k\ln(1/\delta_r)/s}}(1+\frac{e^{\epsilon_l}}{s(e^{\epsilon_l}+1)}))}\]
\end{theorem}
\begin{proof}
	(Sketch)
	Denote $[n_1,\ldots,n_k]$ as numbers of $k$ values in the $i^{\text{th}}$ user's reports set $\mathcal{V}_i$ where $\sum_{i=1}^k n_i=s+1$.
	Denote $\mathcal{S}_i$ (or $\mathcal{S}_i'$) as the dummy dataset generated by the $i^{\text{th}}$ user.
	Let $v$ be the true value held by the $i^{\text{th}}$ user.
	According to the definition of LDP, we compare the probability of shuffler receiving the same set $\mathcal{V}_i$ with numbers of k values are $[n_1,\ldots,n_k]$ when $v=1$ and $v=2$.
	We prove the following formula
	\[\setlength\abovedisplayskip{0.5ex}
	\setlength\belowdisplayskip{0.5ex}
	\begin{aligned}
		\Pr [R(v=1)\cup \mathcal{S}_{i}=\mathcal{V}_{i}]\le e^{\epsilon_r}\Pr [R(v=2)\cup \mathcal{S}_{i}'=\mathcal{V}_{i}]+\delta_r
	\end{aligned}\]
	According to Equation (\ref{eq:rr}) of GRR mechanism, we have
	\[\setlength\abovedisplayskip{0.5ex}
	\setlength\belowdisplayskip{0.5ex}
	\begin{aligned}
		\frac{\Pr [R(v=1)\cup \mathcal{S}_{i}=\mathcal{V}_{i}]}{\Pr [R(v=2)\cup \mathcal{S}_{i}'=\mathcal{V}_{i}]}&\le\frac{e^{\epsilon_l}n_1+n_2}{e^{\epsilon_l}n_2+n_1}\\
	\end{aligned}\]
	Both $n_1$ and $n_2$ follow binomial distribution, where $n_1$ follows $Bin(s,\frac{1}{k})+1$ and $n_2$ follows $Bin(s,\frac{1}{k})$.
	Based on Lemma \ref{lem:assumption}, $n_1$ and $n_2$ are considered as mutually independent variables.
	Then we have
	\[\setlength\abovedisplayskip{0.5ex}
	\setlength\belowdisplayskip{0.5ex}
	\begin{aligned}
		\frac{\Pr [R(v=1)\cup \mathcal{S}_{i}=\mathcal{V}_{i}]}{\Pr [R(v=2)\cup \mathcal{S}_{i}'=\mathcal{V}_{i}]}\le\frac{Bin(s,\frac{1}{k})}{Bin(s,\frac{1}{k})}+\frac{e^{\epsilon_l}}{Bin(s,\frac{1}{k})(e^{\epsilon_l}+1)}\\
	\end{aligned}\]
	According to the Chernoff bound, we have
	\[\setlength\abovedisplayskip{0.5ex}
	\setlength\belowdisplayskip{0.5ex}
	\Pr [Bin(s,\frac{1}{k})\ge\frac{s}{k}(1-\sqrt{\frac{2k\ln(1/\delta_r)}{s}})]\ge 1-\delta_r\]
	Note that $1-\sqrt{\frac{2k\ln(1/\delta_r)}{s}}\ge 0$, so $\frac{k}{s}\le\frac{1}{2\ln(1/\delta_r)}$.
	As $Bin(s,\frac{1}{k})\in[\frac{s}{k}(1-\sqrt{\frac{2k\ln(1/\delta_r)}{s}}),s]$ with the probability larger than $1-\delta_r$, we have
	\[\setlength\abovedisplayskip{0.5ex}
	\setlength\belowdisplayskip{0.5ex}
	\begin{aligned}
		\frac{\Pr [R(v=1)\cup \mathcal{S}_{i}=\mathcal{V}_{i}]}{\Pr [R(v=2)\cup \mathcal{S}_{i}'=\mathcal{V}_{i}]}&\le\frac{k}{1-\sqrt{\frac{2k\ln(1/\delta_r)}{s}}}(1+\frac{e^{\epsilon_l}}{s(e^{\epsilon_l}+1)})\\
	\end{aligned}\]
	when $\frac{k}{s}\le(\sqrt{c^2+\frac{4c}{\sqrt{2\ln(1/\delta_r)}}}-c)^2)$, it has $\frac{k}{1-\sqrt{\frac{2k\ln(1/\delta_r)}{s}}}(1+\frac{e^{\epsilon_l}}{s(e^{\epsilon_l}+1)})\\\le e^{\epsilon_l}$, where $c=\sqrt{2\ln(1/\delta_r)}\cdot\frac{e^\epsilon_l(e^{\epsilon_l}+1)}{s(e^{\epsilon_l}+1)+e^{\epsilon_l}}$.
	
	The details of this proof is in Appendix \ref{app:shuffle_privacy}.
\end{proof}

Figure \ref{fig:shuffle} shows the theoretical result of Theorem \ref{the:shuffle_privacy}.
We set $\epsilon_l=5$, $\delta_r=10^{-4}$, and show the privacy budget $\epsilon_r$ with varying numbers of dummy points on three different sizes of the data domain.
We can observe that dummy points significantly reduce the upper bound of the privacy budget of LDP protection, which enables enhanced privacy protection is provided when the shuffler colludes with the analyst.

\begin{figure}
	\centering
	\includegraphics[width=0.4\textwidth]{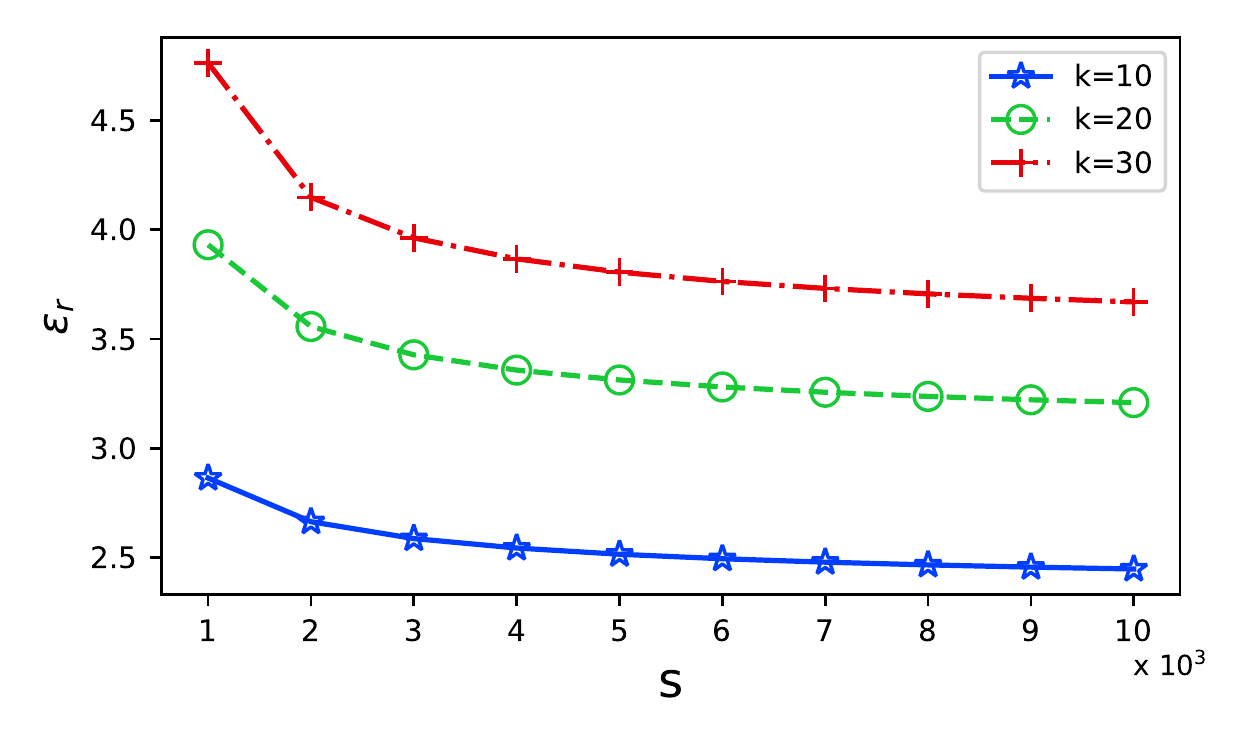}
	\vspace{-3ex}
	\setlength{\belowcaptionskip}{0.5ex} 
	\caption{Enhancement of LDP privacy protection by dummy points, where $\epsilon_l=5$ and $\delta_r=10^{-4}$.}
	\label{fig:shuffle}
\end{figure}

\noindent\textbf{Accuracy Analysis.}
We still use MSE as metrics to analyze the accuracy of mixDUMP. 
Same as before, we first show $\tilde f_v$ is unbiased in mixDUMP in Lemma \ref{lem:mixDUMP}.

\begin{lemma}
	\label{lem:mixDUMP}		
	The estimation result $\tilde f_v$ for any $v\in\mathbb {D}$ in mixDUMP is an unbiased estimation of $f_v$.
\end{lemma}

The proof of Lemma \ref{lem:mixDUMP} is deferred to Appendix \ref{app:mix_unbias}.
Then the MSE of frequency estimation equals to the variance of $\tilde f_v$ in mixDUMP. 
Next, we show the MSE of mixDUMP in Theorem \ref{the:mixuti}.

\begin{theorem}
	\label{the:mixuti}
	The MSE of frequency estimation in mixDUMP is 	
	\[\frac{1}{n}\cdot \frac{e^{\epsilon_l}+k-2}{(e^{\epsilon_l}-1)^2}+\frac{s(k-1)}{n k^2}\cdot (\frac{e^{\epsilon_l}+k-1}{e^{\epsilon_l}-1})^2\]
	
	where $\epsilon_l$ is the privacy budget of GRR in mixDUMP.
\end{theorem}

\begin{proof}
	We give the estimation function in Algorithm \ref{alg:mix_a}. 
	With $p=\frac{e^{\epsilon_l}}{e^{\epsilon_l}+k-1},\ q=\frac{1}{e^{\epsilon_l}+k-1}$ in Equation (\ref{eq:rr}), and $\lambda=kq$, we replace $\lambda$ with $q$ to simplify the derivation. 
	We can get the variance of $\tilde f_v$ is		
	\[\setlength\abovedisplayskip{0.5ex}
	\setlength\belowdisplayskip{0.5ex}
	\begin{aligned}
		Var[\tilde f_v]&=Var[\frac{1}{n}\cdot \frac{\sum_{i\in[n],j\in [s+1]}\mathbbm{1}_{\{v=y_{i,j}\}}-ns\cdot \frac{1}{k}-nq}{p-q}]\\
		& =\frac{1}{n^2}\cdot \frac{1}{(p-q)^2}Var[\sum_{i\in[n],j\in [s+1]}\mathbbm{1}_{\{v=y_{i,j}\}}]\\
		& =\frac{nf_{v}p(1-p)+n(1-f_{v})q(1-q)+ns(k-1)/k^2}{n^2(p-q)^2}\\
		& =\frac{nq(1-q)+nf_{v}(p(1-p)-q(1-q))+ns(k-1)/k^2}{n^2(p-q)^2}\\
	\end{aligned}\]
	As in \cite{wang2017locally}, $f_v$ is assumed to be small on average, then we have
	\[\begin{aligned}
		Var[\tilde f_v]&\simeq\frac{1}{n^2(p-q)^2}(nq(1-q)+\frac{ns(k-1)}{k^2})\\
		& =\frac{1}{n}\cdot \frac{e^{\epsilon_l}+k-2}{(e^{\epsilon_l}-1)^2}+\frac{s(k-1)}{n k^2}\cdot (\frac{e^{\epsilon_l}+k-1}{e^{\epsilon_l}-1})^2\\
	\end{aligned}\]
	According to Equation (\ref{eq:mse}) we have
	\[\setlength\abovedisplayskip{0.5ex}
	\setlength\belowdisplayskip{0.5ex}
	\begin{aligned}
		MSE=&\frac{1}{k}\sum_{v\in \mathbb{D}}E[(\tilde f_v-f_v)^2]=\frac{1}{k}\sum_{v\in \mathbb{D}}Var(\tilde f_v)=\\
		& \frac{1}{n}\cdot \frac{e^{\epsilon_l}+k-2}{(e^{\epsilon_l}-1)^2}+\frac{s(k-1)}{n k^2}\cdot (\frac{e^{\epsilon_l}+k-1}{e^{\epsilon_l}-1})^2\\
	\end{aligned}\]
\end{proof}

\noindent\textbf{Efficiency Analysis.}
Next, we show the efficiency analysis of mixDUMP.
For an input $x_i\in\mathbb{D}$, the output of the user side $M(x_i)$ is $s+1$ messages of length $O(\log k)$ bits.
Since we amplify twice in the process of proving the upper bound of the privacy budget in mixDUMP, the advantage of mixDUMP compared to pureDUMP in the number of messages is weakened.
When mixDUMP satisfies $(\epsilon,\delta)$-DP, $s$ is at most $O(14k\ln(4/\delta)/(n\epsilon^2)-\lambda(1-\sqrt{2\ln(2/\delta)/(\lambda n)}))$.
The upper bound of $s$ can be approximated as $O(14k\ln(4/\delta)/(n\epsilon^2))$ when $\epsilon$ is close to $0$.

The shuffler transmits $n(s+1)$ dummy points to the analyst with the space of $O(ns\log k)$ bits.
The analyst outputs $O(k\log(n(s+1))$ bits estimated from $n(s+1)$ messages with expected MSE is $O(\ln(1/\delta)/(n\epsilon(1-\lambda))^2)$.
Note that $\lambda=k/(e^{\epsilon_l}+k-1)$ is the probability that each user randomizes its data.
As $\epsilon_l\ge \epsilon$, we can amplify $\lambda\le k/(e^{\epsilon}+k-1)$.
Then, the estimation error can be approximated as $O(\frac{1}{\epsilon^2n^2}\cdot(1+\frac{k}{e^{\epsilon}-1})^2\cdot\log\frac{1}{\delta})$.

\subsection{Analysis for Practical Implementations}
\label{sub:flex}
We propose two protocols for privacy-preserving histogram estimation in Subsection \ref{subsec:pureDUMP} and Subsection \ref{subsec:mixDUMP}.
Under the assumption that each user sends $s$ dummy points, we analyze the performance of the proposed protocols theoretically.
In practical implementations, some users might be limited by network bandwidth or computation capability (e.g., running on mobile devices or sensors), and some users may even refuse to send additional dummy points.
Taking these conditions into account, we also analyze the privacy and accuracy of the proposed protocols under a more realistic assumption.

The weak assumption of flexible DUMP protocols assumes that each user sends dummy points with the probability of $\gamma$.
More specifically, each user uniformly samples $s$ random dummy points from $\mathbb{D}$ with the probability of $\gamma$, and only sends one message with the probability of $1-\gamma$.
In practical implementations, $\gamma$ can be obtained by investigating users' desires or referring to the past experience of collectors.
In the following part, we first try to answer how much the weak assumption affects the privacy of protocols in DUMP framework.
We prove the following lemma.

\begin{lemma}
	\label{lem:relax}
	For any DUMP protocols that satisfy $(\epsilon,\ \delta)$-DP for any $n\in\mathbb{N}^+$ and $\gamma\in(0,1]$, 
	the flexible DUMP protocols satisfy $(\epsilon-\ln(1-e^{-\gamma n}),\ \delta+e^{-\gamma n})$-DP.
\end{lemma}

\begin{proof}
	The probability that all users do not generate dummy points is $\Pr [A]=(1-\gamma)^n\le e^{-\gamma n}$, then it has $\Pr [\overline A]\ge 1-e^{-\gamma n}$.
	According to the results of Theorem \ref{the:pureprivacy} and Theorem \ref{the:mixprivacy}, the DUMP protocols satisfy
	\[\setlength\abovedisplayskip{0.5ex}
	\setlength\belowdisplayskip{0.5ex}
	\begin{aligned}
		\Pr [R(\mathcal{D})\cup \mathcal{S}=\mathcal{V}|\overline A]\le e^{\epsilon}\Pr [R(\mathcal{D}')\cup \mathcal{S}'=\mathcal{V}|\overline A]+\delta\\
	\end{aligned}\]
	where $R(\mathcal{D})$ and $R(\mathcal{D}')$ are outputs of data randomizers, $\mathcal{S}$ and $\mathcal{S}'$ are sets of dummy points.
	Then we have 
	\[\setlength\abovedisplayskip{0.5ex}
	\setlength\belowdisplayskip{0.5ex}
	\begin{aligned}
		&\Pr [R(\mathcal{D})\cup \mathcal{S}=\mathcal{V}]\\
		&=\Pr [R(\mathcal{D})\cup \mathcal{S}=\mathcal{V}|\overline A]\Pr [\overline A]+\Pr [R(\mathcal{D})\cup \mathcal{S}=\mathcal{V}|A]\Pr [A]\\
		&\le \Pr [R(\mathcal{D})\cup \mathcal{S}=\mathcal{V}|\overline A]+e^{-\gamma n}\\
		&\le e^{\epsilon}\Pr [R(\mathcal{D}')\cup \mathcal{S}'=\mathcal{V}|\overline A]+\delta+e^{-\gamma n}\\
		&\le \frac{1}{\Pr [\overline A]}\cdot e^{\epsilon}\Pr [R(\mathcal{D}')\cup \mathcal{S}'=\mathcal{V}]+\delta+e^{-\gamma n}\\
		&\le \frac{1}{1-e^{-\gamma n}}\cdot e^{\epsilon}\Pr [R(\mathcal{D}')\cup \mathcal{S}'=\mathcal{V}]+\delta+e^{-\gamma n}\\
	\end{aligned}\]
\end{proof}

According to Lemma \ref{lem:relax}, we can get Theorem \ref{the:relaxpure} and Theorem \ref{the:relaxmix}.
We show that if each user sends dummy points with the probability of $\gamma$, it can increase the privacy budget and the failure probability of DUMP protocols.
Because under this weak assumption, there is a possibility that no user sends dummy points, and it may cause the protocols to fail to satisfy DP.

\begin{theorem}
	\label{the:relaxpure}
	The flexible pureDUMP satisfies $(\epsilon_d-\ln(1-e^{-\gamma n}),\ \delta_d+e^{-\gamma n})$-DP for any $k,S,n\in \mathbb{N}^+,\ \gamma\in(0,1]$,
	where $\delta_d\in(0,0.2907]$, and $\epsilon_d=\sqrt{14k\ln(2/\delta_d)/(S-1)}$.
\end{theorem}

\begin{theorem}
	\label{the:relaxmix}
	The flexible mixDUMP satisfies $(\epsilon_s-\ln(1-e^{-\gamma n}),\ \delta_s+e^{-\gamma n})$-DP for any $k,S,n\in \mathbb{N}^+,\ \lambda\in(0,1],\ \gamma\in(0,1]$,
	where $\delta_s\in(0,0.5814]$, and $\epsilon_s=\sqrt{\frac{14k\ln(4/\delta_s)}{(S+(n-1)\lambda-\sqrt{2(n-1)\lambda\ln(2/\delta_s)}-1)}}$.
\end{theorem}

Next, we analyze the impact of weak assumptions on the accuracy of the DUMP protocols.
Same as before, we first show that the estimation results of flexible pureDUMP and mixDUMP are unbiased in Lemma \ref{lem:relax_pureDUMP} and Lemma \ref{lem:relax_mixDUMP}.
Both the estimation functions and proofs are deferred to Appendix \ref{app:relax_pure_unbias} and Appendix \ref{app:relax_mix_unbias}.

\begin{lemma}
	\label{lem:relax_pureDUMP}		
	The estimation result $\tilde f_v$ for any $v\in\mathbb {D}$ in  flexible pureDUMP is an unbiased estimation of $f_v$.
\end{lemma}

\begin{lemma}
	\label{lem:relax_mixDUMP}		
	The estimation result $\tilde f_v$ for any $v\in\mathbb {D}$ in  flexible mixDUMP is an unbiased estimation of $f_v$.
\end{lemma}

As the weak assumption reduces the strength of privacy protection, the fewer dummy points reduce the MSE of the estimation results.
We show the accuracy of the flexible DUMP protocols in Theorem \ref{the:relax_pureuti} and Theorem \ref{the:relax_mixuti}, and proofs are deferred to Appendix \ref{app:relax_pureuti} and Appendix \ref{app:relax_mixuti}.
\begin{theorem}
	\label{the:relax_pureuti}
	The MSE of frequency estimation in flexible pureDUMP is     
	\[\frac{s\gamma(k-\gamma)}{nk^2}\]
	
\end{theorem}

\begin{theorem}
	\label{the:relax_mixuti}
	The MSE of frequency estimation in flexible mixDUMP is 		
	\[\frac{1}{n}\cdot \frac{e^{\epsilon_l}+k-2}{(e^{\epsilon_l}-1)^2}+\frac{s\gamma(k-\gamma)}{n k^2}\cdot (\frac{e^{\epsilon_l}+k-1}{e^{\epsilon_l}-1})^2\]
	
	where $\epsilon_l$ is the privacy budget of GRR in mixDUMP.
\end{theorem}

\section{Evaluation}

\subsection{Setup}
We design experiments to evaluate the performance of the proposed protocols (pureDUMP and mixDUMP) and compare them with the existing protocols
(truncation-based \cite{balcer2019separating}, private-coin\cite{ghazi2019power}, public-coin \cite{ghazi2019power}, SOLH \cite{wang2020improving}, correlated-noise \cite{ghazi2020private2}).
Note that although the source codes of these existing protocols except SOLH are not formally released, we still implemented them with our best effort and included them for comparison purposes.

In experiments, we measure 
(1) the accuracy of the proposed protocols, i.e., how much accuracy they improve over existing shuffle protocols; 
(2) the communication overhead of the proposed protocols, i.e., how many messages does each user needs to send; and how do they compare with existing shuffle protocols.
(3) how key parameters would affect the accuracy of the proposed protocols, i.e., the number of users $n$ and the size of domain $k$;
Towards these goals, we conduct experiments over both synthetic and real-world datasets. 
The synthetic datasets allow us to adjust the key parameters of the datasets to observe the impact on the utility of protocols.
And the real-world datasets would show the utility of the proposed protocols in a more practical setting.

\begin{figure}[t]
	\centering
	\includegraphics[width=0.5\textwidth]{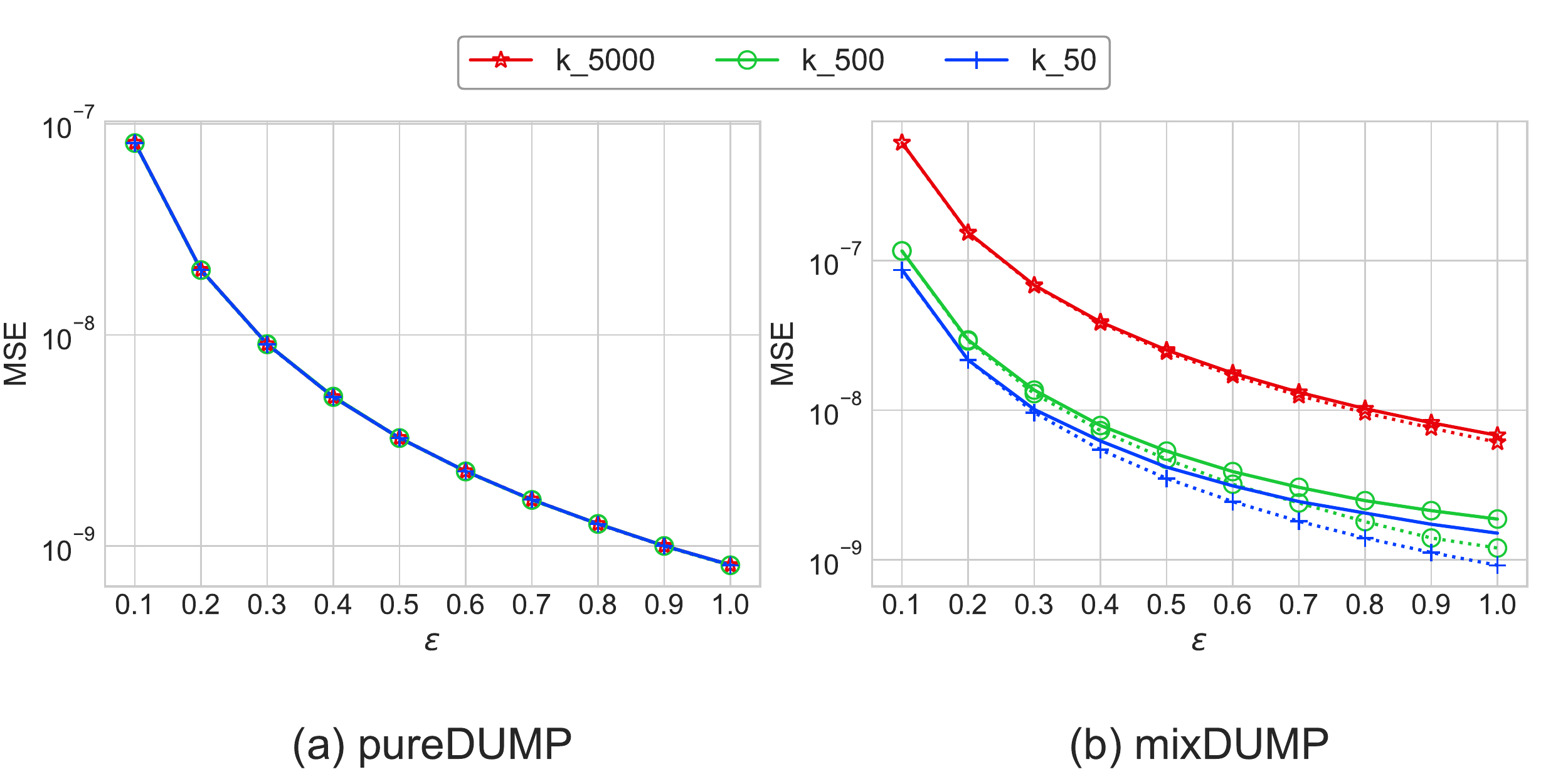}
	\vspace{-5ex}
	\setlength{\belowcaptionskip}{0.5ex}
	\caption{Experimental results of impact of size of data domain on pureDUMP and mixDUMP, where $n$ is fixed to 500,000, $\delta=10^{-6}$. Dashed lines are theoretical results.}
	\label{fig:impact_k}
\end{figure}

\begin{figure}[t]
	\centering
	\includegraphics[width=0.5\textwidth]{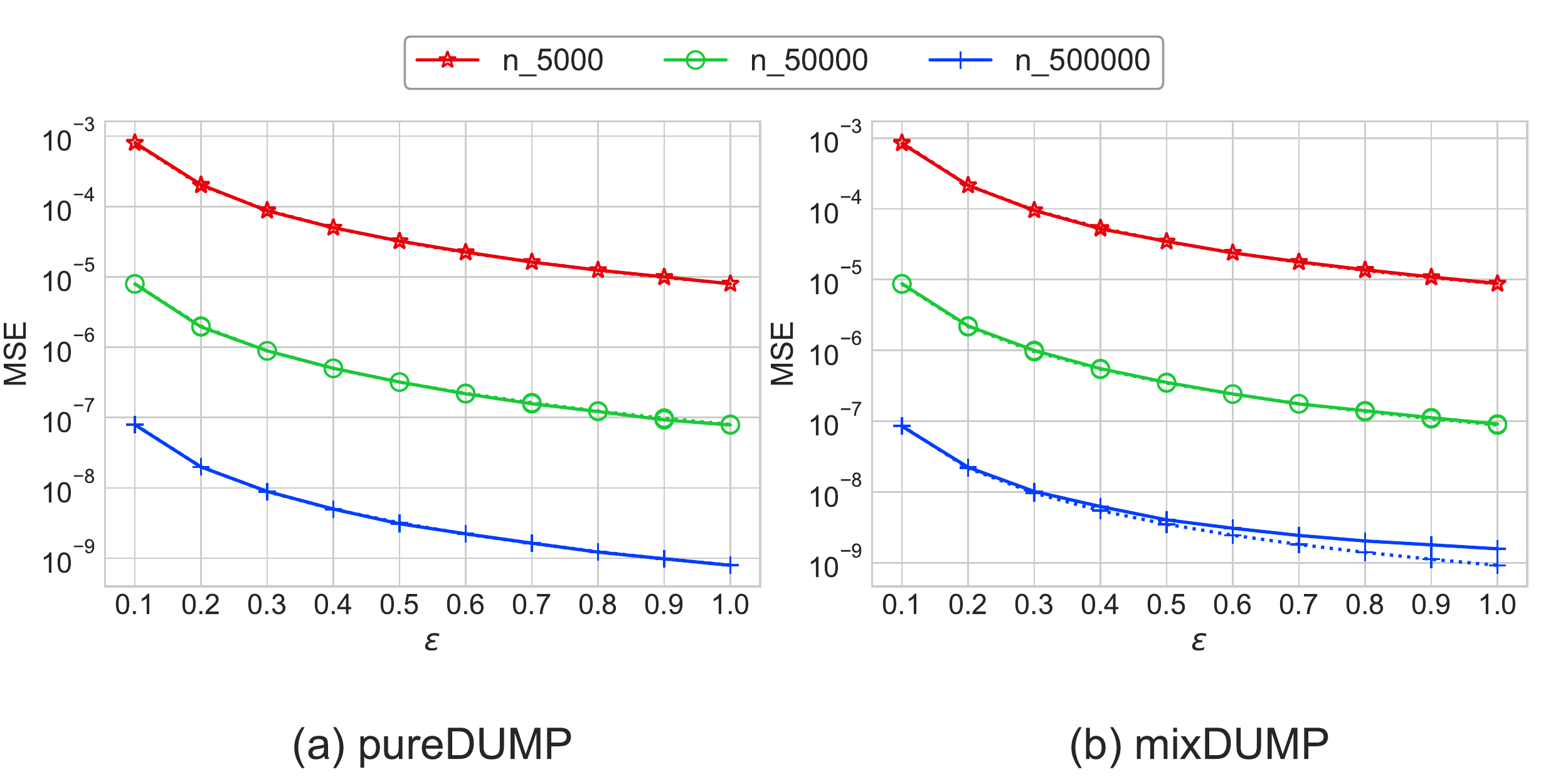}
	\vspace{-5ex}
	\setlength{\belowcaptionskip}{0.5ex}
	\caption{Experimental results of impact of number of users on pureDUMP and mixDUMP, where $k$ is fixed to 50, $\delta=10^{-6}$. Dashed lines are theoretical results.}
	\label{fig:impact_n}
\end{figure}

\noindent\textbf{Environment.}
All protocols are implemented in Python 3.7 with pyhton-xxhash 1.4.4 and numpy 1.18.1 libraries. 
All experiments are conducted on servers running Ubuntu 18.04.4 LTS 2.50GHz with 4*32GB memory.

\noindent\textbf{Synthetic Datasets.}
We generate six synthetic datasets with different sizes of the data domain $k$ and the number of users $n$.
In the first three datasets, we fix $n=500,000$ and generate different datasets with $k=50$, $k=500$, and $k=5,000$, respectively.
In the rest of datasets, we fix $k=50$ and generate different datasets with $n=5,000$, $n=50,000$, and $n=500,000$, respectively.
All datasets are following a uniform distribution.
We observe that the distribution of datasets has little effect on the performance of protocols, so we only show the experimental results on uniform distribution.

\noindent\textbf{Real-World Datasets.}
We conduct experiments on following two real-world datasets.
\begin{compactitem}
	\item IPUMS \cite{ipums-web}: We sample $1\%$ of the US Census data from the Los Angeles and Long Beach area for the year 1970. The dataset contains 70,187 users and 61 attributes. 
	We select the birthplace attribute for evaluation, which includes 900 different values.
	\item Ratings \cite{ratings-web}: This dataset contains 20 million ratings of the movies have been rated by 138,000 users.
	We sample 494,352 records containing 2,000 different movies.
\end{compactitem}

\begin{figure}[t]
	\centering
	\includegraphics[width=0.5\textwidth]{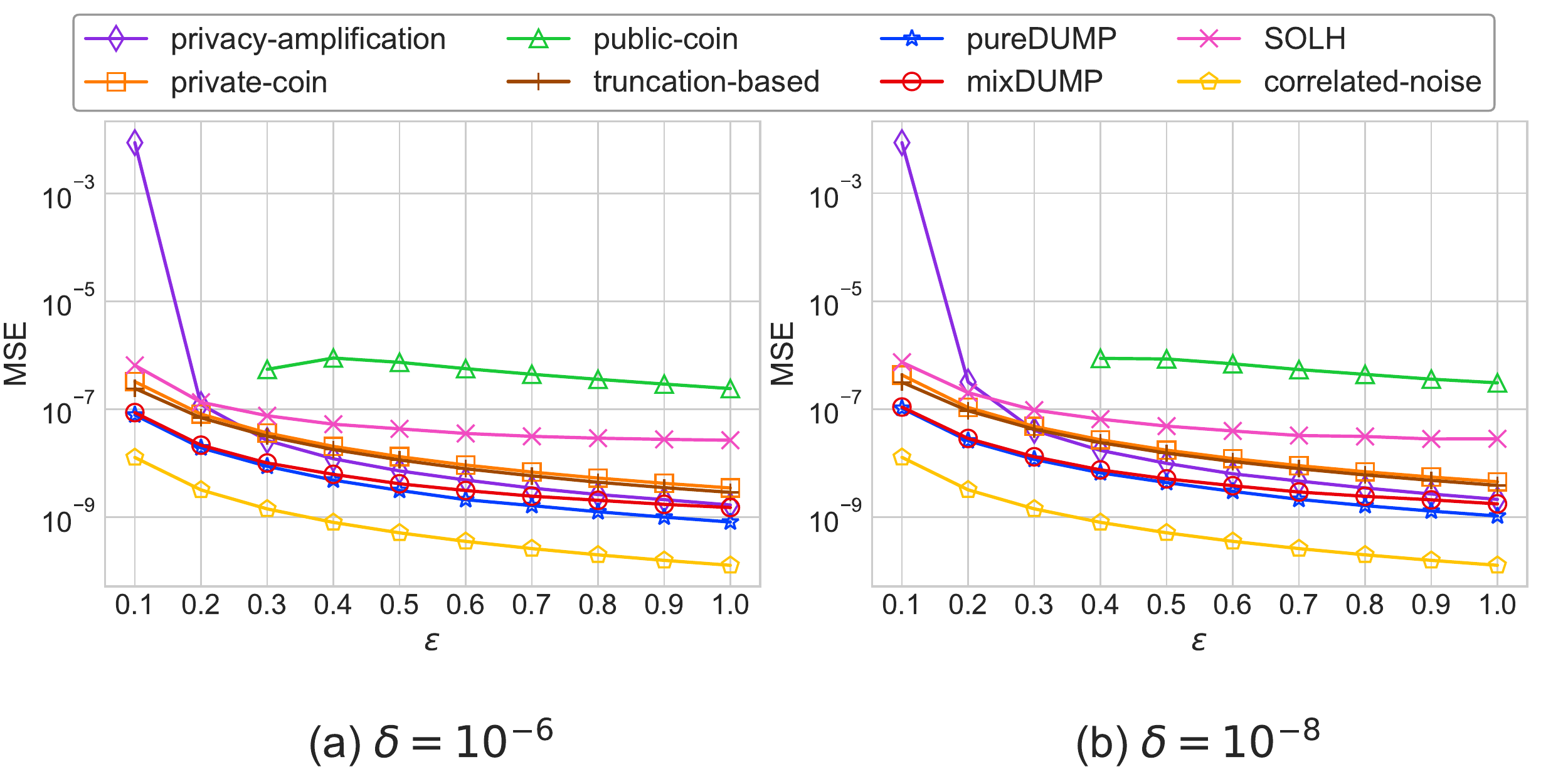}
	\vspace{-5ex}
	\setlength{\belowcaptionskip}{0.5ex}
	\caption{Results of MSE varying $\epsilon$ on the synthetic dataset, where the number of users $n=500,000$, the domain size $k=50$.}
	\label{fig:synk_50}
\end{figure}

\begin{figure}[t]
	\centering
	\includegraphics[width=0.5\textwidth]{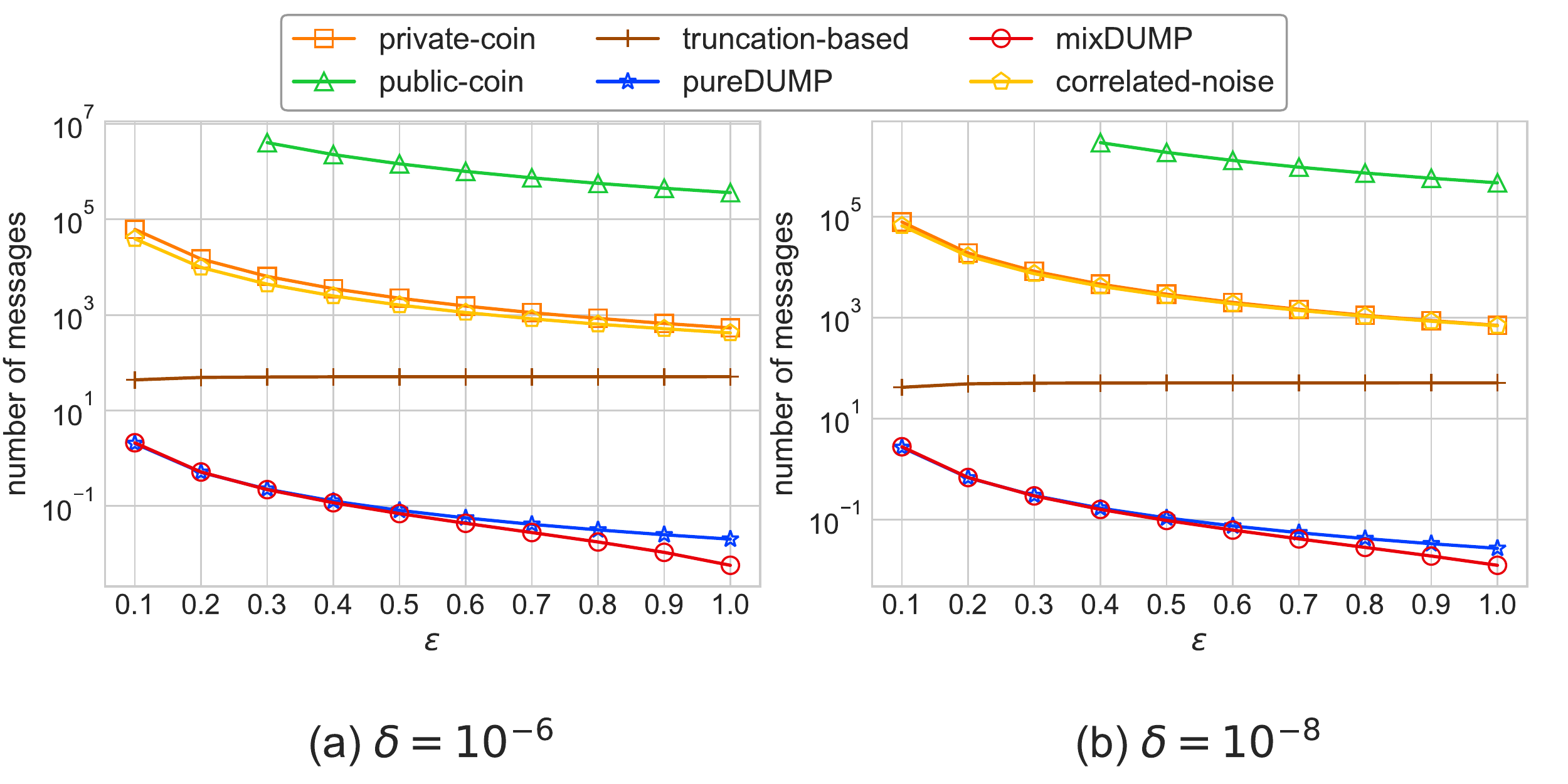}
	\vspace{-5ex}
	\setlength{\belowcaptionskip}{0.5ex}
	\caption{Number of messages sent by each user varying $\epsilon$ on the synthetic dataset, where the number of users $n=500,000$, the domain size $k=50$.}
	\label{fig:synk_50_com}
\end{figure}

\noindent\textbf{Protocols for Comparison.}
We compare pureDUMP and mixDUMP with following protocols.
\begin{compactitem}
	\item private-coin \cite{ghazi2019power}: A \emph{multi-message} protocol for histogram frequency estimation in the shuffle model. 
	It utilizes the Hadamard response as the data randomizer.
	
	\item public-coin \cite{ghazi2019power}: A \emph{multi-message} protocol for histogram frequency estimation in the shuffle model. 
	It is based on combining the Count Min data structure \cite{cormode2005improved} with GRR mechanism.
	It has $\epsilon>\tau\sqrt{90\ln(2\tau/\delta)/n}$ when randomization probability $\gamma$ is set to $90 \tau^2\ln(2\tau/\delta)/(n\epsilon^2)$, where $\tau=\ln(2k)$.
	
	\item SOLH \cite{wang2020improving}: A \emph{single-message} protocol for histogram frequency estimation in the shuffle model.
	It utilizes the hashing-based mechanism OLH as the data randomizer.
	
	\item truncation-based \cite{balcer2019separating}: A \emph{multi-message} protocol for histogram frequency estimation in the shuffle model.
	The protocol can provide privacy protection with $\epsilon>\sqrt{50\ln(2/\delta)/n}$.
	
	\item correlated-noise \cite{ghazi2020private2}: A \emph{multi-message} protocol for histogram frequency estimation in the shuffle model.
	It adds noises from negative binomial distribution to users' data. 
	
	\item privacy-amplification \cite{balle2019privacy}: A privacy amplification result in the shuffle model.
	It has amplification effect only when $\epsilon$ within $(\sqrt{14k\ln{(2/\delta)}/(n-1)},1]$.
\end{compactitem}

\noindent\textbf{Accuracy Metrics.}
In our experiments, we use mean squared error (MSE) as the accuracy metrics.
For each value $v$ in $\mathbb{D}$, we compute the real frequency $f_v$ on the sensitive dataset and estimate frequency $\tilde f_v$ on the noise dataset.
Then we calculate their squared difference for $k$ different values.
\[\setlength\abovedisplayskip{0.5ex}
\setlength\belowdisplayskip{0.5ex}
MSE=\frac{1}{k}\sum_{v\in \mathbb{D}}(\tilde f_v-f_v)^2\]
\noindent All MSE results in experiments are averaged with 50 repeats.

For the private-coin and correlated-noise, we calculate the expected mean squared error.
As $O(nk\log n)$ comparisons are required to match messages with Hadamard codewords on the analyst side of private-coin, which leads to excessive computational overhead, i.e., 1,450 hours for dataset with k=512 and n =70,000.
In the correlated-noise, we didn't find a way to sample the negative binomial distribution $NB(r, q)$ where $r$ is non-integer.
We hope in the further we can further improve this implementation.

\begin{figure}[t]
	\centering
	\includegraphics[width=0.5\textwidth]{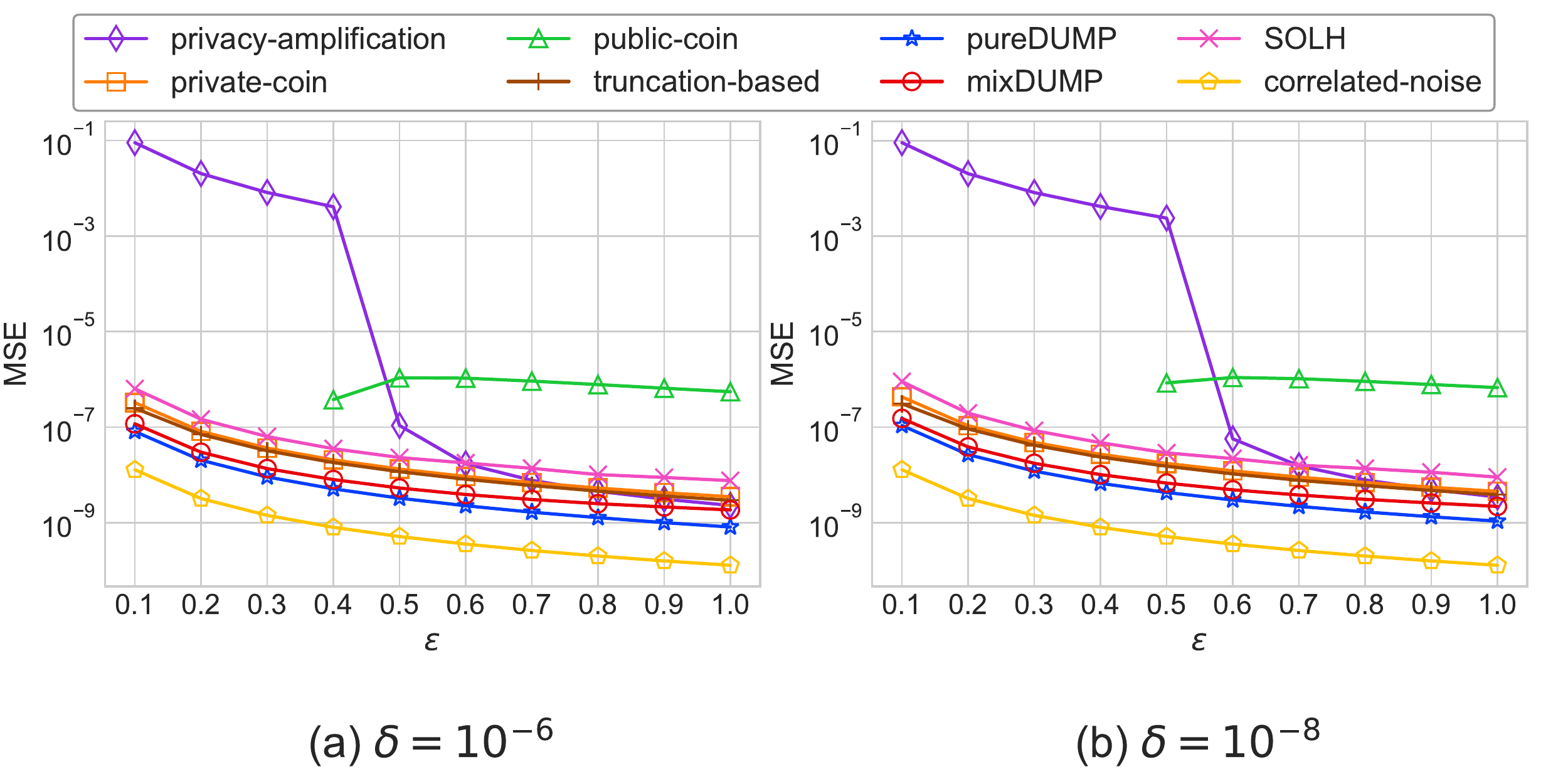}
	\vspace{-5ex}
	\setlength{\belowcaptionskip}{0.5ex}
	\caption{Results of MSE varying $\epsilon$ on the synthetic dataset, where the number of users $n=500,000$, the domain size $k=500$.}
	\label{fig:synk_500}
\end{figure}

\begin{figure}[t]
	\centering
	\includegraphics[width=0.5\textwidth]{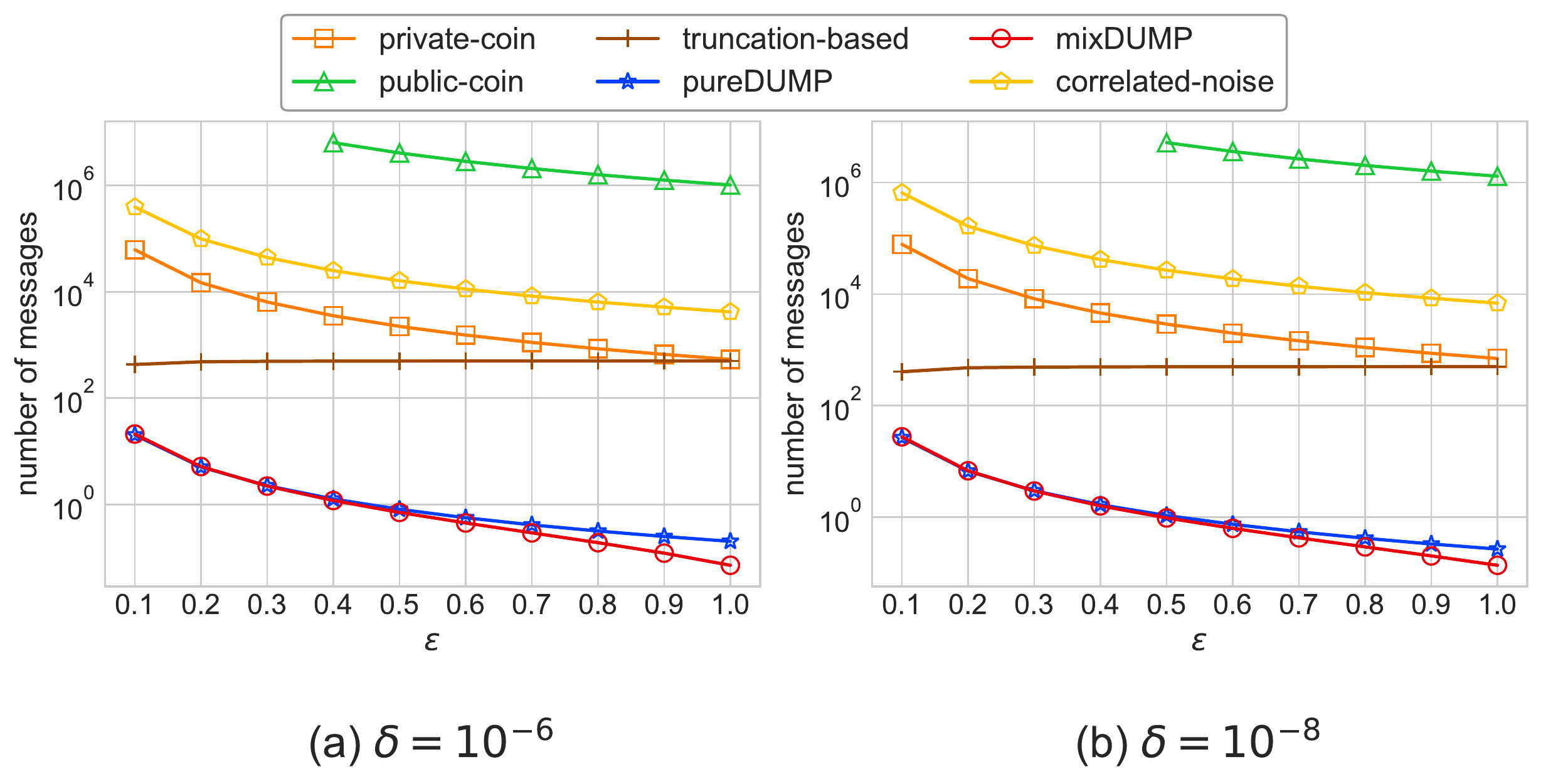}
	\vspace{-5ex}
	\setlength{\belowcaptionskip}{0.5ex}
	\caption{Number of messages sent by each user varying $\epsilon$ on the synthetic dataset, where the number of users $n=500,000$, the domain size $k=500$.}
	\label{fig:synk_500_com}
\end{figure}

\vspace{-1ex}
\subsection{Experiments with Synthetic Datasets}
\noindent\textbf{Overall Results.}
The influence of domain size $k$ on the MSE of pureDUMP and mixDUMP is shown in Figure \ref{fig:impact_k}, where the domain size has almost no influence on the MSE of pureDUMP.
The influence of number of users $n$ on the MSE of pureDUMP and mixDUMP is shown in Figure \ref{fig:impact_n}, where the MSE of pureDUMP and mixDUMP is in inverse proportion to the number of users.
Meanwhile, the small gap between the theoretical and empirical results validates the correctness of our theoretical error analysis. 

Figure \ref{fig:synk_50} and Figure \ref{fig:synk_500} show MSE comparisons of all protocols on two different synthetic datasets.
The proposed protocols, pureDUMP and mixDUMP, support any $\epsilon$ in the range of $(0,1]$ with higher accuracy.
Meanwhile, the numbers of messages sent by each user in all protocols except shuffle-based and SOLH (always one message) are depicted in Figure \ref{fig:synk_50_com} and Figure \ref{fig:synk_500_com}.
The proposed protocols, pureDUMP and mixDUMP, support any $\epsilon$ in the range of $(0,1]$ with the best accuracy compared with other protocols, while the only exception correlated-noise, achieving near-optimal accuracy, but has much higher communication overhead 
\footnote{Since the coefficient $500e^{\epsilon^2/\log (1/\delta)}$ is ignored, the declared asymptotic upper bound is smaller than eMSE.}.
Meanwhile, we observe that the numbers of messages sent by each user in pureDUMP and mixDUMP are always the smallest in all protocols.
Although the MSE of truncation-based is performing similar to pureDUMP and mixDUMP, the communication overhead of private-shuffle is at least one order of magnitude higher than pureDUMP and mixDUMP.
Since sufficient privacy guarantee can be provided by privacy amplification when hundreds of thousands of users are involved, we set $\epsilon_l$ of mixDUMP equals to 8 throughout the experiment to adapt large $\epsilon$ in the range of $(0,1]$.

\begin{figure}[t]
	\centering
	\includegraphics[width=0.5\textwidth]{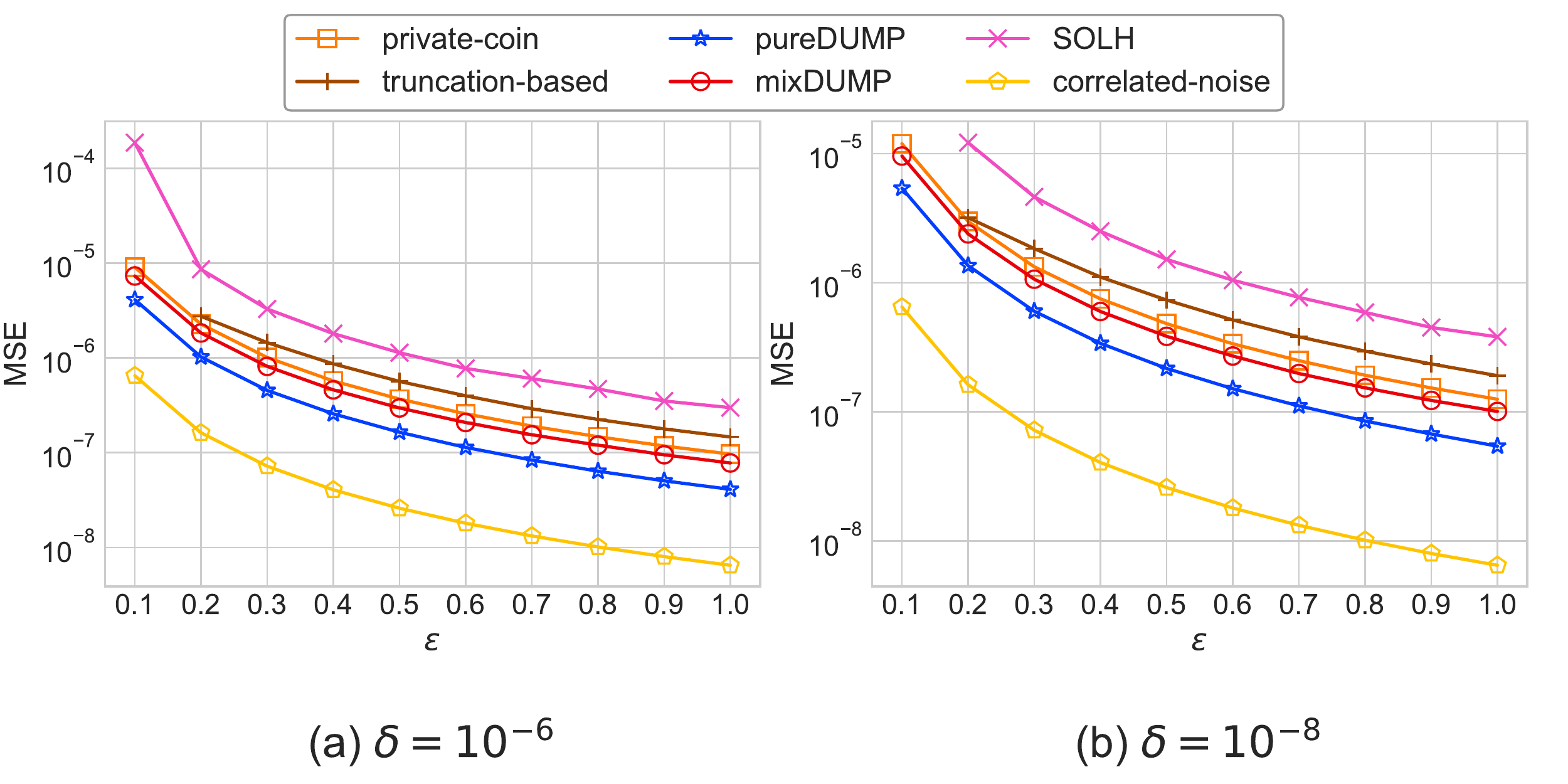}
	\vspace{-5ex}
	\setlength{\belowcaptionskip}{0.5ex}
	\caption{Results of MSE varying $\epsilon$ on the IPUMS dataset.}
	\label{fig:ipums}
\end{figure}

\begin{figure}[t]
	\centering
	\includegraphics[width=0.5\textwidth]{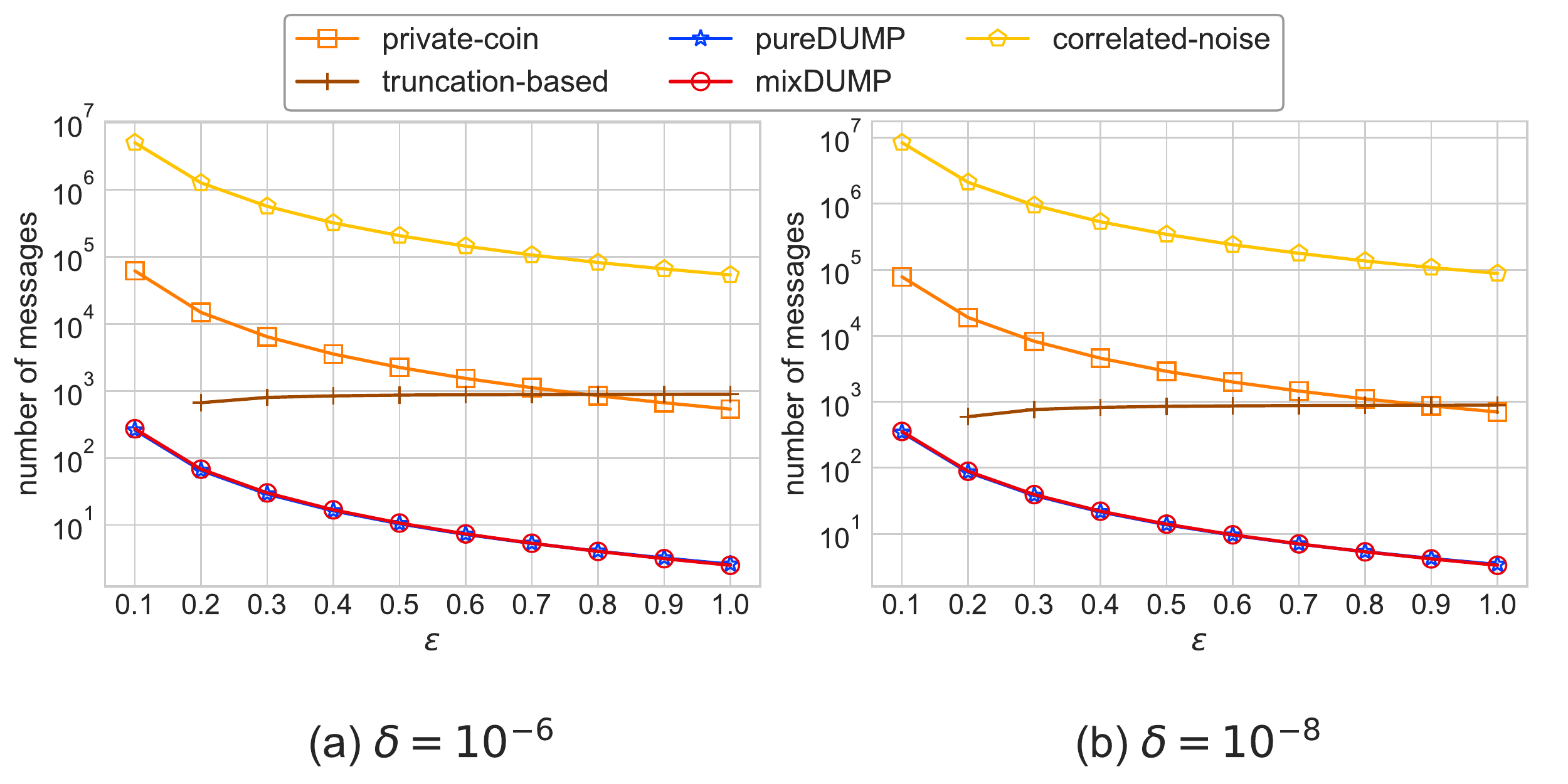}
	\vspace{-5ex}
	\setlength{\belowcaptionskip}{0.5ex}
	\caption{Number of messages sent by each user varying $\epsilon$ on the IPUMS dataset.}
	\label{fig:ipums_com}
\end{figure}

\noindent\textbf{Influence of Domain Size.}
We compare the MSE of pureDUMP on three synthetic datasets with different sizes of data domain in Figure \ref{fig:impact_k}a.
The result shows that the size of the data domain has little influence on the accuracy of pureDUMP.
Because according to our theoretical result in Subsection \ref{subsec:pureDUMP}, the MSE of pureDUMP can be approximately simplified to $14\ln(2/\delta)/(n^2\epsilon^2)$ when $k$ is large.
Figure \ref{fig:impact_k}b shows that the MSE of mixDUMP is proportional to the domain size $k$ because the large size of the data domain increases the error introduced by the data randomizer in mixDUMP.

\noindent\textbf{Influence of Users Number.}
Figure \ref{fig:impact_n}a and Figure \ref{fig:impact_n}b show the influence of the number of users on the MSE of pureDUMP and mixDUMP, respectively.
We can observe that the MSE in inverse proportion to the number of users $n$.
The MSE decreases by two orders of magnitude when the number of users increases by an order of magnitude, which is consistent with our efficiency analysis in Subsection \ref{subsec:pureDUMP} and Subsection \ref{subsec:mixDUMP}.

\begin{figure}[t]
	\centering
	\includegraphics[width=0.5\textwidth]{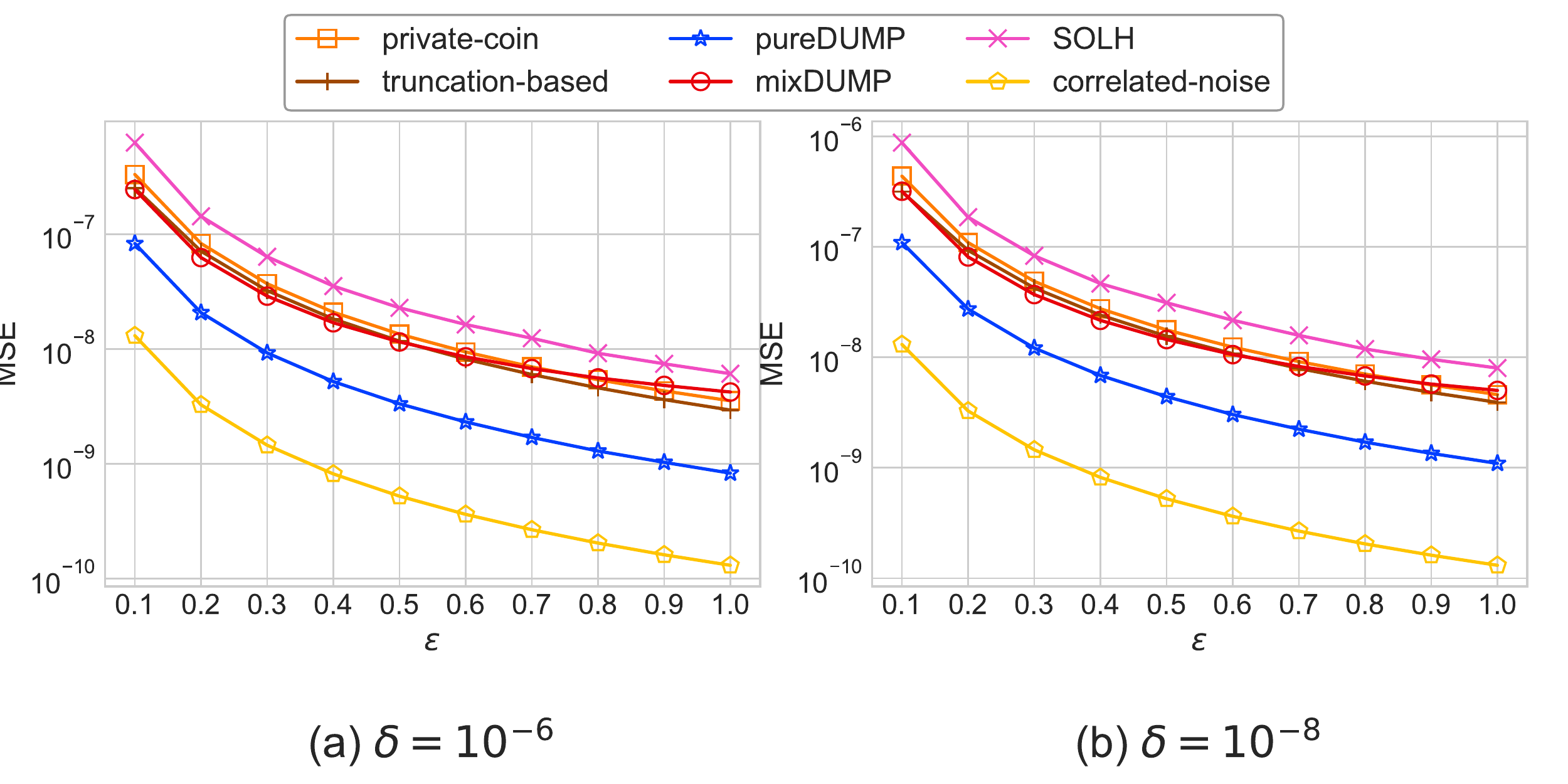}
	\vspace{-5ex}
	\setlength{\belowcaptionskip}{0.5ex}
	\caption{Results of MSE varying $\epsilon$ on the Ratings dataset.}
	\label{fig:ratings}
\end{figure}

\begin{figure}[t]
	\centering
	\includegraphics[width=0.5\textwidth]{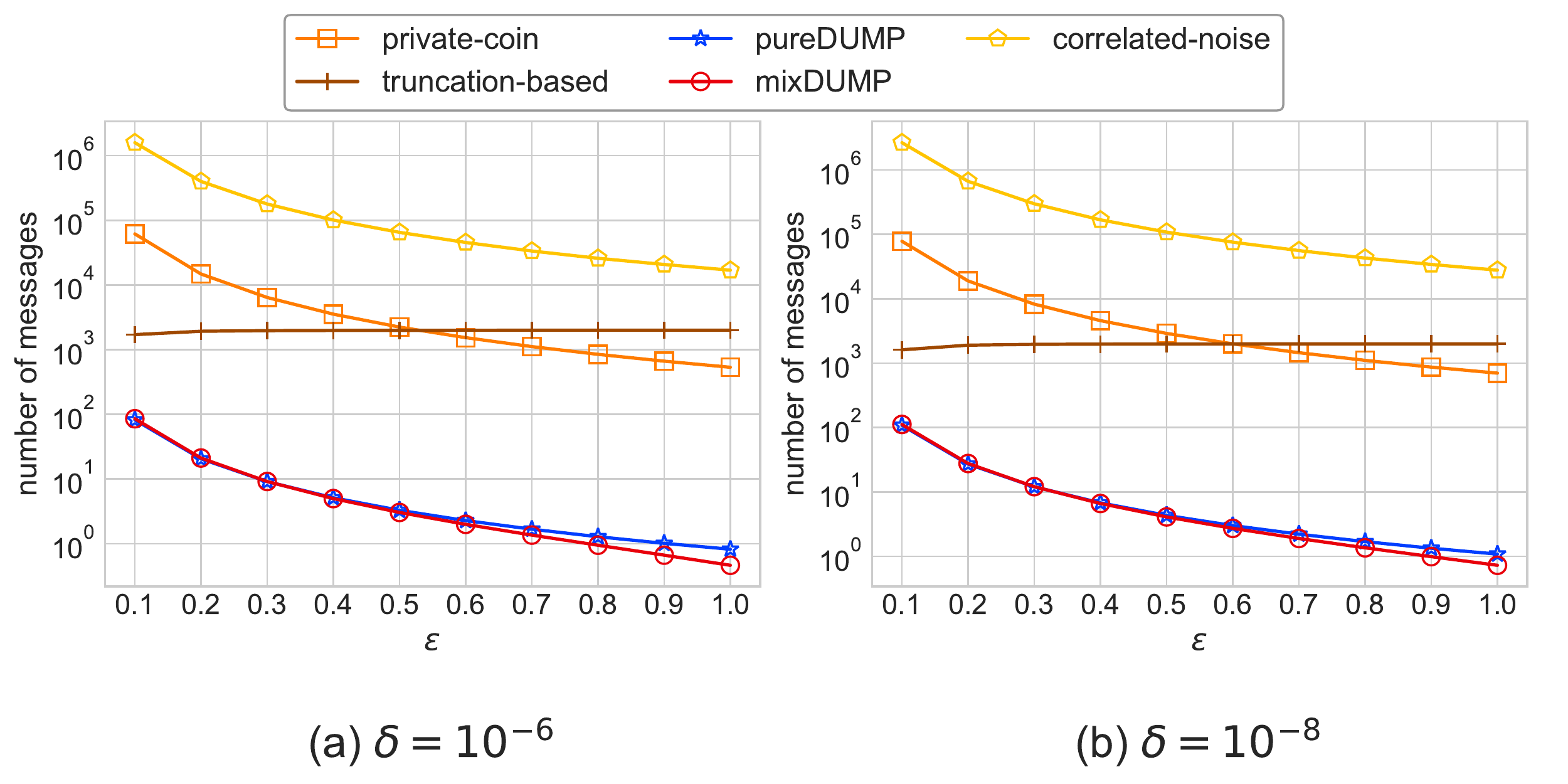}
	\vspace{-5ex}
	\setlength{\belowcaptionskip}{0.5ex}
	\caption{Number of messages sent by each user varying $\epsilon$ on the Ratings dataset.}
	\label{fig:ratings_com}
\end{figure}

\noindent\textbf{Accuracy Comparisons of Existing Protocols.}
We choose two synthetic datasets that enable as many protocols as possible can provide privacy protection with $\epsilon$ within the range of $(0,1]$, and MSE values of some $\epsilon$ out of the valid range of protocols are omitted.
The comparison results is shown in Figure \ref{fig:synk_50}.
We can observe that pureDUMP and mixDUMP always enjoy privacy amplification and have smaller MSE than the other ones except correlated-noise.
Meanwhile, we can also observe that pureDUMP has higher accuracy than mixDUMP.
It means that dummy points can introduce less error than the data randomizer when providing the same strength of privacy protection. 
Besides, the privacy-amplification has no privacy amplification on some small $\epsilon$ values, which causes poor MSE results.
Moreover, we can also observe that the lines of public-coin always slightly increase at the beginning.
Because the noise obeys binomial distribution, i.e., $Bin(n, \gamma)$, where $\gamma$ decreases with the increase of $\epsilon$, the variance increases before $\gamma$ drops to $0.5$.

\noindent\textbf{Communication Comparisons of Existing Protocols.}
The communication overhead is compared by computing the number of messages sent by each user in protocols.
Since the number of messages sent by each user in private-coin and public-coin is always larger than the other protocols and only these two protocols has length of message larger than $O(\log k)$ for datasets with $n>k$ (see Table \ref{tab:compare_pro}), comparing the number of messages does not affect the fairness of the results.
Moreover, as the number of messages sent by each user in privacy-amplification and SOLH is always $1$, we omit them in our comparisons.
We also omit the $\epsilon$ values which are out of the valid range of the protocol.
The comparison results are shown in Figure \ref{fig:synk_50_com}.
We can observe that each user always sends fewer messages than others in pureDUMP and mixDUMP, and the mixDUMP requires the least number of messages per user because the data randomizer provides partial privacy protection so that fewer dummy points need to be generated.
Note that the number of messages that each user sends in DUMP protocols is less than $1$ when a large $\epsilon$ is required.
It means that fewer than $n$ dummy points are needed to satisfy a large $\epsilon$ when there is a sufficient number of users.
According to the theoretical results of flexible DUMP protocols in Subsection \ref{sub:flex}, we provide their parameter settings for different $\epsilon$ on these two synthetic datasets in Appendix \ref{app:relax_s}.
We can also observe that the number of messages sent by each user increases as $\epsilon$ increases in truncation-based protocol.
Because each user sends additional $1$ with probability of $p=1-50\ln(2/\delta)/(\epsilon^2n)$ where $p\in(1/2,1)$.
Hence the probability $p$ closes to 1 when $\epsilon$ increases and a larger number of 1s are sent by each user.
\vspace{-1.5ex}
\subsection{Experiments with Real-World Datasets}
We also conduct experiments on two real-world datasets: IPUMS and Ratings.
The MSE comparisons of protocols are shown in Figure \ref{fig:ipums} and Figure \ref{fig:ratings}, and their corresponding communication overhead are shown in Figure \ref{fig:ipums_com} and Figure \ref{fig:ratings_com}.
The privacy-amplification has no privacy amplification on both IPUMS and Ratings datasets, so it's not involved in the comparison.
Besides, public-coin is also not involved in the comparison because most of $\epsilon\in(0, 1]$ are not in the valid range of public-coin on both datasets.
We can observe that the pureDUMP and mixDUMP still achieve higher accuracy among all protocols except the correlated-noise under the large domain size in real-world datasets.
Over the Ratings dataset with $\epsilon=1$ and $\delta=10^{-6}$, pureDUMP (and mixDUMP) only require around $0.8$ (and $0.5$) expected extra messages per user, whiletruncation-based, private-coin, and correlated-noise need around $10^2$, $10^3$ and $10^4$ extra messages respectively. On the other hand, single-message protocols privacy-amplification and SOLH are less competitive in accuracy and/or subjected to upper-bounds on privacy guarantee (namely, the lower-bounds on $\epsilon$).

\section{Related Works}

\noindent\textbf{Anonymous Data Collection.}
Many works focus on anonymizing users' messages.
Some mechanisms are based on the cryptographic idea, including \cite{syverson2004tor, lazar2018karaoke, tyagi2017stadium, van2015vuvuzela}.
Recently, Chowdhury et al. \cite{roy2020crypt} propose a system \textsf{Crypt$\epsilon$} that achieves the accuracy guarantee of CDP with no trusted data collector required like LDP.
In this approach, users encrypt their messages with homomorphic encryption and then they send messages to an auxiliary server to aggregate and add random noise.
After that, the encrypted aggregated result is sent to the analyst, and the analyst decrypts to get the noisy result that satisfies CDP.
Other mechanisms are based on the shuffling idea. 
This idea is originally proposed by Bittau et al. \cite{bittau2017prochlo}.
They design \textsf{Encode}, \textsf{Shuffle}, \textsf{Analyze} (ESA) architecture.
The ESA architecture eliminates the link between the message and the user identity by introducing a shuffler between users and the analyst.
A lot of following works analyze the privacy benefits that come from the shuffle model \cite{erlingsson2019amplification, cheu2019distributed, balle2019privacy, feldman2020hiding}.
Among them, Feldman et al. \cite{feldman2020hiding} provide the tightest upper bound that $\epsilon$ can be reduced to $(1-e^{-\epsilon})\sqrt{e^{\epsilon}\log (1/\delta)}/\sqrt{n}$.

\noindent\textbf{Protocols in the shuffle model.}
Recent shuffle-related works seek to design more accurate communication-efficient protocols.
The studies focus on frequency estimation in the shuffle model are shown in Section \ref{sub:relta histogram_estimation}.
We show some other related works as a supplement.
Some studies focus on binary data summation.
Cheu et al. \cite{cheu2019distributed} design a \emph{single-message} protocol for binary summation.
Each user only sends one-bit message to the analyst and has expected error at most $O(\sqrt{\log(1/\delta)}/\epsilon)$.
Ghazi et al. \cite{ghazi2020pure} design a \emph{multi-message} protocol for binary summation where each user sends $O(\log n/\epsilon)$ one-bit messages to the analyst and has expected error at most $O(\sqrt{\log(1/\epsilon)/\epsilon^3})$.
Balcer et al. \cite{balcer2019separating} propose a \emph{two-message} protocol for binary summation with expected error at most $O(\log(1/\delta)/\epsilon^2)$.
There are also a lot of studies on real data summation and many excellent protocols are proposed \cite{cheu2019distributed, balle2019privacy, ghazi2019scalable, balle2019differentially, ghazi2020private}.
\section{Conclusions}
In this paper, a novel framework called DUMP is proposed for privacy-preserving histogram estimation in the shuffle model.
Two types of protection mechanisms, data randomization and dummy-point generator, are contained on the user side of DUMP.
The analysis results show that the dummy points can get more privacy benefits from shuffling operation.
To further explore the advantages of dummy points, two protocols: pureDUMP and mixDUMP are proposed.
The proposed protocols achieve better trade-offs between privacy guarantee, statistics accuracy, and communication overhead than existing protocols.
Besides, the proposed protocols are also analyzed under a more realistic assumption for practical application, we prove that the privacy and utility of the proposed protocols still hold.
Finally, we perform experiments to compare different protocols and demonstrate the advantage of the proposed protocols on both synthetic and real-world datasets.
For future work, we will study different distributions of dummy points and apply DUMP to more statistical types such as real number summation, range query, and quantile query.


\bibliographystyle{ACM-Reference-Format}
\bibliography{sample}

\begin{appendix}
\section{Parameter setting in flexible pureDUMP and mixDUMP}
\label{app:relax_s}

Table \ref{tab:relax_pure} and Table \ref{tab:relax_mix} show the number of dummy points need to be sent by a user in flexible pureDUMP and mixDUMP, respectively.
The number of users is $500,000$, we show results on two different $k$ (size of data domain) and $\gamma$ (probability of a user sending dummy points).

\begin{table}[h]
      \centering
      \begin{tabular}{|c|c|c|c|c|}
            \hline
            \multirow{2}*{\diagbox{$\epsilon$}{$k$}}& \multicolumn{2}{|c|}{50}& \multicolumn{2}{|c|}{500}\\
            \cline{2-5}
            & $\gamma=0.01$& $\gamma=0.001$& $\gamma=0.01$& $\gamma=0.001$\\
            \hline
            0.4&13&127&127&1270\\
            \hline
            0.6&6&57&57&565\\
            \hline
            0.8&4&32&32&318\\
            \hline
            1.0&3&21&21&204\\
            \hline
      \end{tabular}
      \setlength{\abovecaptionskip}{1ex} 
      \setlength{\belowcaptionskip}{0.5ex}
      \caption{Number of dummy points sent by a user in flexible pureDUMP. We fix $n=500,000$. All decimals are rounded up.}
      \label{tab:relax_pure}
\end{table}

\begin{table}[h]
      \centering
      \begin{tabular}{|c|c|c|c|c|}
            \hline
            \multirow{2}*{\diagbox{$\epsilon$}{$k$}}& \multicolumn{2}{|c|}{50}& \multicolumn{2}{|c|}{500}\\
            \cline{2-5}
            & $\gamma=0.01$& $\gamma=0.001$& $\gamma=0.01$& $\gamma=0.001$\\
            \hline
            0.4&12&118&119&1190\\
            \hline
            0.6&5&44&46&451\\
            \hline
            0.8&2&18&20&192\\
            \hline
            1.0&1&6&8&72\\
            \hline
      \end{tabular}
      \setlength{\abovecaptionskip}{1ex} 
      \setlength{\belowcaptionskip}{0.5ex}
      \caption{Number of dummy points sent by a user in flexible mixDUMP. We fix $n=500,000$ and $\epsilon_l=8$.All decimals are rounded up.}
      \label{tab:relax_mix}
\end{table}

\section{Proofs}
\subsection{Proof of Lemma \ref{lem:assumption}}
\label{app:assumption}
\begin{proof}
Denote $a_1,a_2,...,a_m$ as $m$ random variables that follow the binomial distribution, where $a_1+a_2+...a_m=N$. 
Let $a_i$ be a constant value $c$, then the conditional probability distribution $Pr[a_j=x|a_i=c]$ ($i\neq j$) can be expressed as
\[\setlength\abovedisplayskip{0.5ex}
\setlength\belowdisplayskip{0.5ex}
\begin{aligned}
    Pr[a_j=x|a_i=c]&=C_{N-c}^x(\frac{1}{m-1})^{x}(\frac{m-2}{m-1})^{N-c-x}\\
    &=\frac{\frac{(N-c)!}{x!(N-c-x)!}(m-2)^{N-c-x}}{(m-1)^{N-c}}\\
\end{aligned}\]
The marginal distribution  $Pr[a_j=x]$ can be expressed as
\[\setlength\abovedisplayskip{0.5ex}
    \setlength\belowdisplayskip{0.5ex}
\begin{aligned}
    Pr[a_j=x]&=C_N^x(\frac{1}{m})^x\cdot(\frac{m-1}{m})^{N-x}\\
    &=\frac{\frac{N!}{x!(N-x)!}\cdot (m-1)^{N-x}}{m^N}\\
\end{aligned}\]
Then we calculate the upper bound of the ratio of these two probability distributions.
\[\setlength\abovedisplayskip{0.5ex}
    \setlength\belowdisplayskip{0.5ex}
\begin{aligned}
    &\frac{Pr[a_j=x|a_i=c]}{Pr[a_j=x]}\\
    &=\frac{\frac{(N-c)!}{x!(N-c-x)!}(m-2)^{N-c-x}}{(m-1)^{N-c}}\cdot \frac{m^N}{\frac{N!}{x!(N-x)!}\cdot (m-1)^{N-x}}\\
    &=\frac{(m-2)^{N-c-x}m^N}{(m-1)^{2N-x-c}}\cdot\frac{(N-c)!(N-x)!}{N!(N-c-x)!}\\
    &=\frac{(m-2)^{N-c-x}}{(m-1)^{N-x-c}}\cdot \frac{m^N}{(m-1)^{N}}\cdot\frac{(N-c)!(N-x)!}{N!(N-c-x)!}\\
    &=(1-\frac{1}{m-1})^{N-c-x}\cdot(1+\frac{1}{m-1})^N\cdot\frac{(N-c)!(N-x)!}{N!(N-c-x)!}\\
    &=(1-(\frac{1}{m-1})^2)^{N-c-x}\cdot(1+\frac{1}{m-1})^{c+x}\cdot\frac{(N-c)!(N-x)!}{N!(N-c-x)!}\\
    &\le (\frac{m}{m-1})^{c+x}\cdot \frac{(N-x)\cdot(N-x-1)\cdot...\cdot(N-x-c+1)}{N\cdot(N-1)\cdot ...\cdot(N-c+1)}\\
    &\le(\frac{m}{m-1})^{c+x}\\
\end{aligned}\]
The lower bound of the ratio of these two probability distributions can be calculated as
\[\setlength\abovedisplayskip{0.5ex}
     \setlength\belowdisplayskip{0.5ex}
\begin{aligned}
    &\frac{Pr[a_j=x|a_i=c]}{Pr[a_j=x]}\\
    &=\frac{(m-2)^{N-c-x}m^N}{(m-1)^{2N-x-c}}\cdot\frac{(N-c)!(N-x)!}{N!(N-c-x)!}\\
    &=\frac{(m-2)^{N-c-x}}{(m-1)^{N-x-c}}\cdot \frac{m^N}{(m-1)^{N}}\cdot\frac{(N-c)!(N-x)!}{N!(N-c-x)!}\\
    &=(1-\frac{1}{m-1})^{N-c-x}\cdot(1+\frac{1}{m-1})^N\cdot \frac{(N-c)!(N-x)!}{N!(N-c-x)!}\\
    &\ge (\frac{m-2}{m-1})^{N-c-x}\cdot \frac{(N-x)\cdot(N-x-1)\cdot...\cdot(N-x-c+1)}{N\cdot(N-1)\cdot ...\cdot(N-c+1)}\\
    &\ge (\frac{m-2}{m-1})^{N-x-c}\cdot(\frac{N-x-c+1}{N-c+1})^c\\
\end{aligned}\]
Since $a_1, a_2,...,a_m$ are assumed to be much smaller than $N$, $m$ can be inferred to be a large value.
Hence both  $\frac{m}{m-1}$ and $\frac{m-2}{m-1}$ are close to $1$. 
Besides, $\frac{N-x-c+1}{N-c+1}$ is also close to $1$ since $x$ is a value much smaller than $N$. 
Therefore, both the upper bound and the lower bound of $\frac{Pr[a_j=x|a_i=c]}{Pr[a_j=x]}$ are close to $1$. 
We have
\[\setlength\abovedisplayskip{0.5ex}
    \setlength\belowdisplayskip{0.5ex}
\begin{aligned}
    Pr[a_j=x, a_i=c]&=Pr[a_j=x|a_i=c]Pr[a_i=c]\\
    &\simeq Pr[a_j=x]Pr[a_i=c]
\end{aligned}\]
We can approximately consider that the distributions of $a_i$ and $a_j$ are independent based on the above proof.

\end{proof}

\subsection{Proof of Theorem \ref{the:pureprivacy}}
\label{app:pureprivacy}
\begin{proof}
    For any two neighboring datasets $\mathcal{D}$ and $\mathcal{D'}$, w.l.o.g., we assume they differ in the $n^\text{th}$ value, and $x_n=1$ in $\mathcal{D}$, $x'_n=2$ in $\mathcal{D'}$.
	Let $\mathcal{S}$ and $\mathcal{S'}$ be two dummy datasets with $|\mathcal{S}|=|\mathcal{S}'|=ns$ dummy points sent by users.
	What the analyst receives is a mixed dataset $\mathcal{V}$ that is a mixture of user dataset $\mathcal{D}$(or $\mathcal{D'})$ and the dummy dataset $\mathcal{S}$(or $\mathcal{S'})$. 
	We denote $\mathcal{V}=\mathcal{D}\cup\mathcal{S}$(or $\mathcal{V}=\mathcal{D'}\cup\mathcal{S'}$).
	Denote $s_j$ as the number of value $j$ in the dummy dataset.
	To get the same $\mathcal{V}$, there should be a dummy point equals to 2 in $\mathcal{S}$ but equals to 1 in $\mathcal{S'}$.
	Thus the numbers of $k$ values in other $|\mathcal{S}|-1$ dummy points are $[s_1-1,s_2,\ldots,s_k]$ in $\mathcal{S}$ and $[s_1,s_2-1,\ldots,s_k]$ in $\mathcal{S'}$.
	Then we have

     \[\setlength\abovedisplayskip{0.5ex}
     \setlength\belowdisplayskip{0.5ex}
     \begin{aligned}
     \frac{\Pr [\mathcal{D}\cup\mathcal{S}=\mathcal{V}]}{\Pr [\mathcal{D'}\cup\mathcal{S'}=\mathcal{V}]}&=\frac{(_{s_1-1,s_2,\ldots,s_k}^{\ \ \ \ \ |\mathcal{S}|-1})}{(_{s_1,s_2-1,\ldots,s_k}^{\ \ \ \ \ |\mathcal{S}|-1})}
     =\frac{C_{s_1+s_2-1}^{s_2}C_{s_1-1}^{s_1-1}}{C_{s_1+s_2-1}^{s_2-1}C_{s_1}^{s_1}}\\
     & =\frac{s_1!(s_2-1)!}{(s_1-1)!s_2!}=\frac{s_1}{s_2}\\
     \end{aligned}\]

    As dummy points are all uniformly random sampled from $\mathbb{D}$, both $s_1$ and $s_2$ follow binomial distributions, where  $s_1$ follows $Bin(|\mathcal{S}|-1,\frac{1}{k})+1$ and $s_2$ follows $Bin(|\mathcal{S}|-1,\frac{1}{k})$.
	Denote $S_1$ as $Bin(|\mathcal{S}|-1,\frac{1}{k})+1$ and $S_2$ as $Bin(|\mathcal{S}|-1,\frac{1}{k})$.
	Here, $s_1$ and $s_2$ can be considered as mutually independent variables according to Lemma \ref{lem:assumption}.
	Then the probability that pureDUMP does not satisfy $\epsilon_d$-DP is
	\[\setlength\abovedisplayskip{0.5ex}
      \setlength\belowdisplayskip{0.5ex}
      \begin{aligned}
            \Pr [\frac{\Pr [\mathcal{D}\cup\mathcal{S} = \mathcal{V}]}{\Pr [\mathcal{D'}\cup\mathcal{S'} = \mathcal{V}]}\ge e^{\epsilon_d}]=\Pr [\frac{S_1}{S_2} \geq e^{\epsilon_d}] \leq \delta_d
      \end{aligned}\]

     Let $c=E(S_2)=\frac{|\mathcal{S}|-1}{k}$, $\Pr [\frac{S_1}{S_2}\ge e^{\epsilon_d}]$ can be divided into two cases with overlap that $S_1\ge ce^{\epsilon_d/2}$ or $S_2\le ce^{-\epsilon_d/2}$ so we have
     \[\setlength\abovedisplayskip{0.5ex}
     \setlength\belowdisplayskip{0.5ex}
     \begin{aligned}
     & \Pr [\frac{S_1}{S_2}\ge e^{\epsilon_d}]\\
     & \le \Pr [S_1\ge ce^{\epsilon_d/2}]+\Pr [S_2\le ce^{-\epsilon_d/2}]\\
     & =\Pr [S_2\ge ce^{\epsilon_d/2}-1]+\Pr [S_2\le ce^{-\epsilon_d/2}]\\
     & =\Pr [S_2-c\ge ce^{\epsilon_d/2}-c-1]+\Pr [S_2-c\le ce^{-\epsilon_d/2}-c]\\
     & =\Pr [S_2-c\ge c(e^{\epsilon_d/2}-1-1/c)]+\Pr [S_2-c\le c(e^{-\epsilon_d/2}-1)]\\
     \end{aligned}\]

     According to multiplicative Chernoff bound, we have
     \[\Pr [\frac{S_1}{S_2}\ge e^{\epsilon_d}]\le e^{-\frac{c}{3}(e^{\epsilon_d}-1-1/c)^2}+e^{-\frac{c}{2}(1-e^{-\epsilon_d/2})^2}\]
   
     To ensure that $\Pr [\frac{S_1}{S_2}\ge e^{\epsilon_d}]\le \delta_d$, we need
     \[e^{-\frac{c}{3}(e^{\epsilon_d}-1-1/c)^2}\le \frac{\delta_d}{2}\ and\ e^{-\frac{c}{2}(1-e^{-\epsilon_d/2})^2}\le \frac{\delta_d}{2}\]

     For $e^{-\frac{c}{2}(1-e^{-\epsilon_d/2})^2}\le \frac{\delta_d}{2}$, $1-e^{-\epsilon_d/2}\ge (1-e^{-1/2})\epsilon_d\ge\epsilon_d/\sqrt{7}$, where $\epsilon_d\le 1$.
     So we can get $c\ge \frac{14\ln{(2/\delta_d)}}{\epsilon_d^2}$.
     For $e^{-\frac{c}{3}(e^{\epsilon_d}-1-1/c)^2}\le \frac{\delta_d}{2}$, $e^{\epsilon_d/2}-1\ge \frac{\epsilon_d}{2}$, so $e^{\epsilon_d/2}-1-1/c \ge \frac{25}{54}\epsilon_d$, where $c\ge 27/\epsilon_d$.
     Therefore, we can get $c\ge \frac{3\cdot 54^2\ln{(2/\delta_d)}}{25^2\epsilon_d^2}$. 
     Note that $\frac{3\cdot 54^2}{25^2}\le 14$.

     In summary, we can get $c\ge max\{\frac{14\ln{(2/\delta_d)}}{\epsilon_d^2},\frac{27}{\epsilon_d}\}$.
     Because $\epsilon_d\le 1$, when $\delta_d\le 0.2907$, $\epsilon_d\le \frac{14\ln{(2/\delta_d)}}{27}$ can be held all the time.
     So $c\ge \frac{14\ln{(2/\delta_d)}}{\epsilon_d^2}$. Therefore, we have $\frac{S-1}{k}\ge \frac{14\ln{(2/\delta_d)}}{\epsilon_d^2}$, 
     is that $\epsilon_d\le \sqrt{\frac{14k\ln{(2/\delta_d)}}{S-1}}$.
\end{proof}

\subsection{Proof of Theorem \ref{the:mixprivacy}}
\label{app:mixprivacy}
\begin{proof}
      The assumptions are same to the proof of Theorem \ref{the:pureprivacy}.
      The $n^\text{th}$ value $x_n=1$ in $\mathcal{D}$ and $x'_n=2$ in its neighboring dataset $\mathcal{D'}$.
      If the $n^{\text{th}}$ user randomizes its data, then we can easily get
     \[\frac{\Pr [R(\mathcal{D})\cup \mathcal{S}=\mathcal{V}]}{\Pr [R(\mathcal{D})\cup \mathcal{S}'=\mathcal{V}]}=1\]
     So, we focus on the case that the $n^{\text{th}}$ user does not randomize its data.
	Denote $\mathcal{H}$ as the set of users participating in randomization in Algorithm \ref{alg:mix_r}, and the size of $\mathcal{H}$ is $|\mathcal{H}|$.
	Since the randomized users' data and dummy points are all follows the uniform distribution, mixDUMP can be regarded as pureDUMP with dummy dataset $\mathcal{H}\cup\mathcal{S}$ is introduced.
	From Theorem \ref{the:pureprivacy}, mixDUMP satisfies $(\epsilon_H,\frac{\delta_s}{2})$-DP, we have
     \[\setlength\abovedisplayskip{0.5ex}
     \setlength\belowdisplayskip{0.5ex}
     \begin{aligned}
           \Pr [R(\mathcal{D})\cup \mathcal{S}=\mathcal{V}]\le e^{\epsilon_H}\Pr [R(\mathcal{D'})\cup \mathcal{S'}=\mathcal{V}]+\frac{\delta_s}{2}
     \end{aligned}\]
     where  $\epsilon_H=\sqrt{\frac{14k\ln{(4/\delta_s)}}{|\mathcal{H}|+|\mathcal{S}|-1}}$, $\delta_s\le 0.5814$.
     Since $\lambda$ is the probability that each user randomizes its data, the number of users participating in randomization $|\mathcal{H}|$ follows $Bin(n-1,\lambda)$. 
     Define $\mu=(n-1)\lambda$, and $\gamma=\sqrt{\frac{2\ln{(2/\delta_s)}}{(n-1)\lambda}}$. 
     According to Chernoff bound, we can get
     \[\Pr [|\mathcal{H}|<(1-\gamma)\mu]<e^{-\mu\gamma^2/2}=\frac{\delta_s}{2}\]
     Then we have
     \[\begin{aligned}
          &\Pr [R(\mathcal{D})\cup \mathcal{S}=\mathcal{V}]\\
          &=\Pr [R(\mathcal{D})\cup \mathcal{S}=\mathcal{V} \cap |\mathcal{H}|\ge (n-1)\lambda-\sqrt{2(n-1)\lambda \ln{(2/\delta_s)}}]+\\
          & \Pr [R(\mathcal{D})\cup \mathcal{S}=\mathcal{V} \cap |\mathcal{H}| < (n-1)\lambda-\sqrt{2(n-1)\lambda \ln{(2/\delta_s)}}]\\
          & \le \Pr [R(\mathcal{D})\cup \mathcal{S}=\mathcal{V} \cap |\mathcal{H}|\ge (n-1)\lambda-\sqrt{2(n-1)\lambda \ln{(2/\delta_s)}}]+\frac{\delta_s}{2}\\
          & Let\ t=(n-1)\lambda-\sqrt{2(n-1)\lambda ln(2/\delta_s)}\\
          & \Pr [R(\mathcal{D})\cup \mathcal{S}=\mathcal{V}]\le(\sum_{h\ge t} \Pr [R(\mathcal{D})\cup \mathcal{S}=\mathcal{V}]\Pr [|\mathcal{H}|=h])+\frac{\delta_s}{2}\\
          & \le(\sum_{h\ge t} (e^{\epsilon_H}\Pr [R(\mathcal{D}')\cup \mathcal{S}'=\mathcal{V}]\Pr [|\mathcal{H}|=h]+\frac{\delta_s}{2}\Pr [|\mathcal{H}|=h]))+\frac{\delta_s}{2}\\
          & \le(\sum_{h\ge t} (e^{\epsilon_H}\Pr [R(\mathcal{D}')\cup \mathcal{S}'=\mathcal{V}]\Pr [|\mathcal{H}|=h])+\delta_s\\
          & \le max_{h\ge t}e^{\epsilon_H}\Pr [R(\mathcal{D}')\cup \mathcal{S}'=\mathcal{V}]+\delta_s\\
          & \le e^{\sqrt{\frac{14k\ln{(4/\delta_s)}}{t+S-1}}}\Pr [R(\mathcal{D}')\cup \mathcal{S}'=\mathcal{V}]+\delta_s\\
          & =e^{\sqrt{\frac{14k\ln{(4/\delta_s)}}{(n-1)\lambda-\sqrt{2(n-1)\lambda \ln{(2/\delta_s)}}+S-1}}}\Pr [R(\mathcal{D}')\cup \mathcal{S}'=\mathcal{V}]+\delta_s\\
     \end{aligned}\]
\end{proof}

\subsection{Proof of Lemma \ref{lem:mixDUMP}}
\label{app:mix_unbias}
\begin{proof}
$f_v$ is the true frequency of element $v$ appears in users' dataset. 
And $\tilde f_v$ is the estimation of $f_v$. 
We prove $\tilde f_v$ in mixDUMP is unbiased.

In GRR mechanism, the probability that each value keeps the true value with the probability $p$, and changes to other values in the data domain with the probability of $q$. 
We give the estimation function of mixDUMP in Algorithm \ref{alg:mix_a}. The formula is 
\[\tilde f_v=\frac{\sum_{i\in[n],j\in [s+1]}\mathbbm{1}_{\{v=y_{i,j}\}}-n\lambda(1-\frac{1}{k})-\frac{ns}{k}}{n(1-2\lambda(1-\frac{1}{k}))}\]
where $\lambda=\frac{k}{k-1}q$. To simplify the derivation, we replace $\lambda$ with $q$, then we have estimation of $\tilde f_v$
\[\begin{aligned}
E[\tilde f_v]&=E[\frac{1}{n}\cdot \frac{(\sum_{i\in[n],j\in [s+1]}\mathbbm{1}_{\{v=y_{i,j}\}}-nq-\frac{ns}{k})}{(1-2q)}]\\
& =\frac{1}{n}\cdot E[\frac{(\sum_{i\in[n],j\in [s+1]}\mathbbm{1}_{\{v=y_{i,j}\}}-nq-\frac{ns}{k})}{(1-2q)}]\\
& =\frac{1}{n}\cdot \frac{1}{1-2q}\cdot E[\sum_{i\in[n],j\in [s+1]}\mathbbm{1}_{\{v=y_{i,j}\}}-nq-\frac{ns}{k}]\\
& =\frac{1}{n}\cdot \frac{1}{1-2q}\cdot (E[\sum_{i\in[n],j\in [s+1]}\mathbbm{1}_{\{v=y_{i,j}\}}]-nq-\frac{ns}{k})\\
& =\frac{1}{n}\cdot \frac{1}{1-2q}\cdot (n f_v(1-q)+n(1-f_v)q+\frac{ns}{k}\\
& -nq-\frac{ns}{k})\\
& =\frac{1}{n}\cdot \frac{1}{1-2q}\cdot (1-2q)\cdot n f_v\\
& =f_v\\
\end{aligned}\]
\end{proof}

\subsection{Proof of Theorem \ref{the:shuffle_privacy}}
\label{app:shuffle_privacy}
\begin{proof}
    Denote $[n_1,\ldots,n_k]$ as numbers of $k$ values in the $i^{\text{th}}$ user's reports set $\mathcal{V}_i$ where $\sum_{i=1}^k n_i=s+1$.
	Denote $\mathcal{S}_i$ (or $\mathcal{S}_i'$) as the dummy dataset generated by the $i^{\text{th}}$ user.
	Let $v$ be the true value held by the $i^{\text{th}}$ user.
	According to the definition of LDP, we compare the probability of shuffler receiving the same set $\mathcal{V}_i$ with numbers of k values are $[n_1,\ldots,n_k]$ when $v=1$ and $v=2$.
	We prove the following formula
      \[\setlength\abovedisplayskip{0.5ex}
      \setlength\belowdisplayskip{0.5ex}
      \Pr [R(v=1)\cup \mathcal{S}_{i}=\mathcal{V}_{i}]\le e^{\epsilon_r}\Pr [R(v=2)\cup \mathcal{S}_{i}'=\mathcal{V}_{i}]+\delta_r\]

      According to Equation (\ref{eq:rr}) of GRR mechanism, we have
      \[\setlength\abovedisplayskip{0.5ex}
      \setlength\belowdisplayskip{0.5ex}
      \begin{aligned}
            &\frac{\Pr [R(v=1)\cup \mathcal{S}_{i}=\mathcal{V}_{i}]}{\Pr [R(v=2)\cup \mathcal{S}_{i}'=\mathcal{V}_{i}]}\\
            &=\frac{e^{\epsilon_l}(_{n_1-1,n_2,\ldots,n_k}^{\ \ \ \ \ \ \ \ s})+(_{n_1,n_2-1,\ldots,n_k}^{\ \ \ \ \ \ \ \ s})+\ldots+(_{n_1,n_2,\ldots,n_k-1}^{\ \ \ \ \ \ \ \ s})}{e^{\epsilon_l}(_{n_1,n_2-1,\ldots,n_k}^{\ \ \ \ \ \ \ \ s})+(_{n_1-1,n_2,\ldots,n_k}^{\ \ \ \ \ \ \ \ s})+\ldots+(_{n_1,n_2,\ldots,n_k-1}^{\ \ \ \ \ \ \ \ s})}\\     
            &\le \frac{e^{\epsilon_l}(_{n_1-1,n_2,\ldots,n_k}^{\ \ \ \ \ \ \ \ s})+(_{n_1,n_2-1,\ldots,n_k}^{\ \ \ \ \ \ \ \ s})}{e^{\epsilon_l}(_{n_1,n_2-1,\ldots,n_k}^{\ \ \ \ \ \ \ \ s})+(_{n_1-1,n_2,\ldots,n_k}^{\ \ \ \ \ \ \ \ s})}
            =\frac{e^{\epsilon_l}n_1+n_2}{e^{\epsilon_l}n_2+n_1}\\
      \end{aligned}\]
      Since dummy points are uniformly random sampled from $\mathbb{D}$, both $n_1$ and $n_2$ follow binomial distribution.
      $n_1$ follows $Bin(s,\frac{1}{k})+1$ and $n_2$ follows $Bin(s,\frac{1}{k})$.
      Based on Lemma \ref{lem:assumption}, $n_1$ and $n_2$ are considered as mutually independent variables.
      We have 
      \[\setlength\abovedisplayskip{0.5ex}
      \setlength\belowdisplayskip{0.5ex}
      \begin{aligned}
        \frac{\Pr [R(v=1)\cup \mathcal{S}_{i}=\mathcal{V}_{i}]}{\Pr [R(v=2)\cup \mathcal{S}_{i}'=\mathcal{V}_{i}]}&\le\frac{e^{\epsilon_l}Bin(s,\frac{1}{k})+e^{\epsilon_l}+Bin(s,\frac{1}{k})}{e^{\epsilon_l}Bin(s,\frac{1}{k})+Bin(s,\frac{1}{k})+1}\\
        &\le\frac{Bin(s,\frac{1}{k})(e^{\epsilon_l}+1)+e^{\epsilon_l}}{Bin(s,\frac{1}{k})(e^{\epsilon_l}+1)}\\
        &=\frac{Bin(s,\frac{1}{k})}{Bin(s,\frac{1}{k})}+\frac{e^{\epsilon_l}}{Bin(s,\frac{1}{k})(e^{\epsilon_l}+1)}\\
      \end{aligned}\]
      
      According to the Chernoff bound $\Pr [Bin(s,\frac{1}{k})<(1-\gamma)\mu]<\delta$, where $\gamma=\sqrt{\frac{2k\ln(1/\delta)}{s}}$ and $\mu=\frac{s}{k}$.
      That is $\Pr [Bin(s,\frac{1}{k})\ge\frac{s}{k}(1-\sqrt{\frac{2k\ln(1/\delta)}{s}})]\ge 1-\delta$.
      Note that $1-\sqrt{\frac{2k\ln(1/\delta)}{s}}\ge 0$, so $\frac{k}{s}\le\frac{1}{2\ln(1/\delta)}$.
      As we prove that $Bin(s,\frac{1}{k})\in[\frac{s}{k}(1-\sqrt{\frac{2k\ln(1/\delta)}{s}},s]$ with the probability larger than $1-\delta$, then we have
      \[\begin{aligned}
        &\frac{\Pr [R(v=1)\cup \mathcal{S}_{i}=\mathcal{V}_{i}]}{\Pr [R(v=2)\cup \mathcal{S}_{i}'=\mathcal{V}_{i}]}\\
        &\le\frac{s}{\frac{s}{k}(1-\sqrt{\frac{2k\ln(1/\delta)}{s}})}+\frac{e^{\epsilon_l}}{\frac{s}{k}(1-\sqrt{\frac{2k\ln(1/\delta)}{s}})(e^{\epsilon_l}+1)}\\
        &=\frac{k}{1-\sqrt{\frac{2k\ln(1/\delta)}{s}}}(1+\frac{e^{\epsilon_l}}{s(e^{\epsilon_l}+1)})\\
      \end{aligned}\]
      when $\frac{k}{s}\le(\sqrt{c^2+\frac{4c}{\sqrt{2\ln(1/\delta)}}}-c)^2)$, it has $\frac{k}{1-\sqrt{\frac{2k\ln(1/\delta)}{s}}}(1+\frac{e^{\epsilon_l}}{s(e^{\epsilon_l}+1)})\le e^{\epsilon_l}$, where $c=\sqrt{2\ln(1/\delta)}\cdot\frac{e^\epsilon_l(e^{\epsilon_l}+1)}{s(e^{\epsilon_l}+1)+e^{\epsilon_l}}$.
\end{proof}

\subsection{Proof of Lemma \ref{lem:relax_pureDUMP}}
\label{app:relax_pure_unbias}
\begin{proof}
      $f_v$ is the true frequency of element $v$ appears in users' dataset. 
      $\tilde f_v$ is the estimation of $f_v$. 
      We prove $\tilde f_v$ in flexible pureDUMP is unbiased.
      
      As each user generates dummy points with probability of $\gamma$, and with probability of $1-\gamma$ does not generate any dummy point, the estimation function is
      \[\tilde f_v=\frac{1}{n}\cdot (\sum_{i\in[n],j\in [s+1]}\mathbbm{1}_{\{v=y_{i,j}\}}-\frac{\gamma ns}{k})\]
      
      Then we have
      \[\begin{aligned}
      E[\tilde f_v]&=E[\frac{1}{n}\cdot (\sum_{i\in[n],j\in [s+1]}\mathbbm{1}_{\{v=y_{i,j}\}}-\frac{\gamma ns}{k})]\\
      & =\frac{1}{n}\cdot E[\sum_{i\in[n],j\in [s+1]}\mathbbm{1}_{\{v=y_{i,j}\}}]-\frac{\gamma s}{k}\\
      & =\frac{1}{n}\cdot(nf_v+\frac{\gamma ns}{k})-\frac{\gamma s}{k}\\
      & =f_v\\
      \end{aligned}\]
\end{proof}

\subsection{Proof of Lemma \ref{lem:relax_mixDUMP}}
\label{app:relax_mix_unbias}
\begin{proof}
      $f_v$ is the true frequency of element $v$ appears in users' dataset. 
      And $\tilde f_v$ is the estimation of $f_v$. 
      We prove $\tilde f_v$ of flexible mixDUMP is unbiased.
      
      In GRR mechanism, the probability that each value keeps unchanged with the probability $p$, and changes to other values in the data domain with the probability of $q$.
      As each user generates dummy points with probability of $\gamma$, and with probability of $1-\gamma$ does not generate any dummy point, the estimation function is 
      \[\tilde f_v=\frac{\sum_{i\in[n],j\in [s+1]}\mathbbm{1}_{\{v=y_{i,j}\}}-n\lambda(1-\frac{1}{k})-\frac{\gamma ns}{k}}{n(1-2\lambda(1-\frac{1}{k}))}\]
      where $\lambda=\frac{k}{k-1}q$. 
      To simplify the derivation, we replace $\lambda$ with $q$, then we have estimation of $\tilde f_v$     
      \[\begin{aligned}
      E[\tilde f_v]&=E[\frac{1}{n}\cdot \frac{(\sum_{i\in[n],j\in [s+1]}\mathbbm{1}_{\{v=y_{i,j}\}}-nq-\frac{\gamma ns}{k})}{(1-2q)}]\\
      & =\frac{1}{n}\cdot E[\frac{(\sum_{i\in[n],j\in [s+1]}\mathbbm{1}_{\{v=y_{i,j}\}}-nq-\frac{\gamma ns}{k})}{(1-2q)}]\\
      & =\frac{1}{n}\cdot \frac{1}{1-2q}\cdot E[\sum_{i\in[n],j\in [s+1]}\mathbbm{1}_{\{v=y_{i,j}\}}-nq-\frac{\gamma ns}{k}]\\
      & =\frac{1}{n}\cdot \frac{1}{1-2q}\cdot (E[\sum_{i\in[n],j\in [s+1]}\mathbbm{1}_{\{v=y_{i,j}\}}]-nq-\frac{\gamma ns}{k})\\
      & =\frac{1}{n}\cdot \frac{1}{1-2q}\cdot (n f_v(1-q)+n(1-f_v)q+\frac{\gamma ns}{k}\\
      & -nq-\frac{\gamma ns}{k})\\
      & =\frac{1}{n}\cdot \frac{1}{1-2q}\cdot (1-2q)\cdot n f_v\\
      & =f_v\\
      \end{aligned}\]
\end{proof}

\subsection{Proof of Theorem \ref{the:relax_pureuti}}
\label{app:relax_pureuti}
\begin{proof}
      As each user sends dummy points with the probability of $\gamma$, the distribution of number of any distinct value in dummy points is $Bin(ns,\frac{\gamma}{k})$.
      Thus, the estimation function of flexible pureDUMP is
      \[\tilde f_v=\frac{1}{n}\cdot(\sum_{i\in[n],j\in [s+1]}\mathbbm{1}_{\{v=y_{i,j}\}}-\frac{\gamma ns}{k})\]

      Then, we have
      \[\begin{aligned}
            Var[\tilde f_v]&=Var[\frac{1}{n}\cdot(\sum_{i\in[n],j\in [s+1]}\mathbbm{1}_{\{v=y_{i,j}\}}-\frac{\gamma ns}{k})]\\
            & =\frac{1}{n^2}\cdot Var[\sum_{i\in[n],j\in [s+1]}\mathbbm{1}_{\{v=y_{i,j}\}}]\\
            & =\frac{1}{n^2}\cdot ns\cdot\frac{\gamma}{k}\cdot{1-\frac{\gamma}{k}}\\
            & =\frac{s\gamma(k-\gamma)}{nk^2}\\
      \end{aligned}\]

      Lemma \ref{lem:relax_pureDUMP} proves that $\tilde f_v$ is unbiased in flexible pureDUMP, thus the MSE of flexible pureDUMP is\\
      \[\setlength\abovedisplayskip{0.5ex}
      \setlength\belowdisplayskip{0.5ex}
      \begin{aligned}
            MSE&= \frac{1}{k}\sum_{v\in \mathbb{D}}E[(\tilde f_v-f_v)^2]=\frac{1}{k}\sum_{v\in \mathbb{D}}Var(\tilde f_v)\\
            & =\frac{s\gamma(k-\gamma)}{nk^2}\\
      \end{aligned}\]
\end{proof}

\subsection{Proof of Theorem \ref{the:relax_mixuti}}
\label{app:relax_mixuti}
\begin{proof}
As each user generates dummy points in probability of $\gamma$ in flexible mixDUMP, the distribution of number of any distinct value in dummy points is $Bin(ns,\frac{\gamma}{k})$.
Thus, the estimation function of flexible mixDUMP is 
\[\tilde f_v=\frac{\sum_{i\in[n],j\in [s+1]}\mathbbm{1}_{\{v=y_{i,j}\}}-n\lambda(1-\frac{1}{k})-\frac{\gamma ns}{k}}{n(1-2\lambda(1-\frac{1}{k}))}\]
where $\lambda=\frac{k}{k-1}q$. 
To simplify derivation, we replace $\lambda$ with $q$, and with $p=\frac{e^{\epsilon_l}}{e^{\epsilon_l}+k-1},\ q=\frac{1}{e^{\epsilon_l}+k-1}$, we have the variance of $\tilde f_v$ is
\[\begin{aligned}
Var[\tilde f_v]&=Var[\frac{1}{n}\cdot \frac{\sum_{i\in[n],j\in [s+1]}\mathbbm{1}_{\{v=y_{i,j}\}}-\gamma ns\cdot \frac{1}{k}-nq}{p-q}]\\
& =\frac{1}{n^2}\cdot \frac{1}{(p-q)^2}Var[\sum_{i\in[n],j\in [s+1]}\mathbbm{1}_{\{v=y_{i,j}\}}]\\
& \simeq\frac{1}{n^2}\cdot \frac{1}{(p-q)^2}(nq(1-q)+ns\cdot \frac{\gamma}{k}\cdot \frac{k-\gamma}{k})\\
& =\frac{1}{n}\cdot \frac{e^{\epsilon_l}+k-2}{(e^{\epsilon_l}-1)^2}+\frac{s\gamma (k-\gamma)}{n k^2}\cdot (\frac{e^{\epsilon_l}+k-1}{e^{\epsilon_l}-1})^2\\
\end{aligned}\]
Lemma \ref{lem:relax_mixDUMP} proves that $\tilde f_v$ is unbiased in flexible mixDUMP, thus we have the MSE of flexible mixDUMP is
\[\setlength\abovedisplayskip{0.5ex}
\setlength\belowdisplayskip{0.5ex}
\begin{aligned}
MSE&= \frac{1}{k}\sum_{v\in \mathbb{D}}E[(\tilde f_v-f_v)^2]=\frac{1}{k}\sum_{v\in \mathbb{D}}Var(\tilde f_v)\\
& =\frac{1}{n}\cdot \frac{e^{\epsilon_l}+k-2}{(e^{\epsilon_l}-1)^2}+\frac{s\gamma (k-\gamma)}{n k^2}\cdot (\frac{e^{\epsilon_l}+k-1}{e^{\epsilon_l}-1})^2\\
\end{aligned}\]
\end{proof}
\end{appendix}
\end{document}